\documentclass[pdflatex,sn-mathphys-ay]{sn-jnl}% Math and Physical Sciences Author Year Reference Style
\theoremstyle{thmstyleone}%
\newtheorem{theorem}{Theorem}%  meant for continuous numbers
\newtheorem{proposition}{Proposition}% to get separate numbers for theorem and proposition etc.

\theoremstyle{thmstyletwo}%

\theoremstyle{thmstylethree}%
\newtheorem{definition}{Definition}%

\usepackage{graphicx} % Required for inserting images
\usepackage{amsmath, amsfonts, amssymb, amsthm}
\usepackage{derivative}
\usepackage{natbib}
\usepackage{xcolor}
\usepackage{enumitem}
\usepackage{booktabs}
\usepackage{siunitx} 

\renewcommand{\Re}{\operatorname{Re}}
\renewcommand{\Im}{\operatorname{Im}}

\title{Conservative deterministic Markov models in mathematical biology: uniqueness of steady states, reversibility and computational methods}
\date{2026}

\author*[1]{\fnm{Joseph G.} \sur{Shuttleworth}}\email{joey.shuttleworth@nottingham.ac.uk}
\author[1]{\fnm{Simon P.} \sur{Preston}}
\author[2,3]{\fnm{Chon Lok} \sur{Lei}}
\author[1]{\fnm{Etienne} \sur{Farcot}}
\author[1]{\fnm{Gary R.} \sur{Mirams}}

\affil[1]{\orgdiv{Centre for Mathematical Medicine and Biology, School of Mathematical Sciences}, \orgname{University of Nottingham}, \orgaddress{\city{Nottingham}, \postcode{NG7 2RD}, \state{} \country{United Kingdom}}}
\affil[2]{\orgdiv{Institute of Translational Medicine, Faculty of Health Sciences}, \orgname{University of Macau}, \orgaddress{ \state{Macau}, \country{China}}}
\affil[3]{\orgdiv{Department of Biomedical Sciences, Faculty of Health Sciences}, \orgname{University of Macau}, \orgaddress{\state{Macau}, \country{China}}}

\abstract{
    Ordinary differential equations are commonly used throughout the sciences to build mechanistic models of time-dependent processes. Often, such models are Markov models describing the time-evolution of different interconnected ``states''. When these models have no ``sources'' or ``sinks'', they naturally conserve the total population of the system. 
    When each state is reachable (directly or indirectly) from any other state, these models are called irreducible.
    For many applications, a deterministic system of ordinary differential equations (ODE) is the most suitable modelling approach. 
    We summarise important mathematical results which show that this irreducibility property guarantees the existence and uniqueness of global stable equilibria, and discuss the computationally-efficient implementation of such ODE-based models. 
    We also discuss the condition of microscopic reversibility and show how it guarantees non-oscillatory behaviour, which enables additional efficiencies in computations. 
    These properties and methods are demonstrated through example models of biological phenomena, where we demonstrate their importance for efficient model fitting and simulation.
}

\begin{document}
\maketitle
\section{Introduction} \label{sec:introduction}
Mathematical models are commonly used to describe time-varying phenomena throughout mathematical biology.  
Mathematical models in which the evolution of the system in question is dependent only on the instantaneous state of the system, and not past behaviour, are said to satisfy the \emph{Markov property} \citep{kelly_reversibility_1981}. 
Whilst such models may be formulated probabilistically, deterministic approximations can often be suitable. 
This may be true when the model describes a large number of individuals which are each subject to the same dynamics, for example \cite{gillespie_deterministic_2009}.
Here, we focus on the properties of deterministic Markov models and provide two case studies concerning mathematical biology models.

Though relevant in some mathematical biology contexts \citep{kepler_stochasticity_2001, gupta_stochastic_2011, clancy_na_2002, clerx_four_2019}, basic results on the properties of such models are often not discussed in detail, or are discussed within a probabilistic setting \citep{lin_efficient_2018}. 
Results concerning deterministic formulations are covered in general terms and technical detail in existing literature \citep{BermanAbraham1979Nmit, seneta06}, but, authors typically discuss broader classes of problems and the relevant results are not always contextualised with examples pertaining to mathematical biology. 

In this article, we present a small collection of essential methods and properties, which prove to be important in cell-biology applications. For example, the method for steady-state computation presented by \citet{fink_markov_2009} relies on \emph{microscopic reversibility}, a property satisfied by the certain models that they discuss. However, some models used in practice do not satisfy this property, and we give a general derivation of a method that can be applied to any relevant model, irrespective of microscopic reversibility.
Similarly, \citet{keener2009mathematical} discuss Markov models in the context of cellular electrophysiology and introduce specific examples and discuss large-time behaviour. 
We provide a more general discussion, 
exploring the existence and uniqueness of Markov models more generally, covering the aforementioned models discussed in \citet{keener2009mathematical}, and any other model satisfying the conditions outlined below.

Our primary focus is the class of ODE-based models which approximate finite-state \emph{continuous-time Markov chains} (CTMCs) \citep{kelly_reversibility_1981, seneta06, wilkinson2018stochastic}. 
These CTMCs describe a single entity stochastically transitioning between some collection of states at randomly distributed times, where the distribution of waiting times and subsequent state-to-state transition is dependent only on the current state. 
Assuming a finite collection of states, we consider an \(N \times N\) \emph{transition-rate matrix}, \(\mathbf{Q}\), where \(N\) is the number of states, and for \(i, j \leqslant N\) with \(i\neq j\), the \(j^\text{th}\) entry in the \(i^\text{th}\) row is \(Q_{i, j}\)---the transition rate from state \(i\) to state \(j\).
For all examples provided in this article, the matrix \(\mathbf{Q}\) is either constant, or time-dependent such that for any \(i, j\) either \(Q_{i,j}(t) = 0\) for all \(t\) or \(Q_{i, j}(t) > 0\) for all \(t\). 
In this way, the \emph{model topology} (the state-to-state connections described by the model) does not change over time.

The results discussed in Section~\ref{sec:connectedness} concern only the model topology, and are not dependent on particular values of transition rates (besides them being always positive, or \(0\)).
Therefore, we omit time-inhomogeneity for now, and describe the average behaviour of many individuals subject to these state-to-state transitions with the equation,
\begin{equation}
    \odv{\mathbf{x}}{t} = \mathbf{Q}^\top \mathbf{x}\,,\label{eqn:governing_equation}
\end{equation}
where \(\mathbf{x}\) represents the amount of individuals in each state of the system. If \(\mathbf{x}(t)\) is interpreted as a probability vector describing the probability that a randomly-transitioning individual is in each state at time \(t\) (necessitating that these probabilities sum to \(1\)), Equation~\eqref{eqn:governing_equation} is exactly the Kolmogorov forward equation of the equivalent CTMC \citep{kelly_reversibility_1981}.
%---though in terms of column vectors instead of row vectors. 

Nevertheless, the above governing equation is a purely deterministic description of how the state of a system evolves over time \citep{keener2009mathematical}. 
In some contexts, such as ion-channel models in electrophysiology, these models are referred to as \emph{Markov models} \citep{rudy_computational_2006, fink_markov_2009, teed_computationally_2016-1}. And we use the same term throughout this article.

The number of states, the value of transition rates, and the {model topology} may all vary between models.
These properties are encoded in the transition-rate matrix, \(\mathbf{Q}\), which necessarily satisfies the following properties (where \(Q_{i, j}\) denotes the element on the \(i\)\textsubscript{th} row and \(j\)\textsuperscript{th} column of \(\mathbf{Q}\)):
\begin{align}
&\mathbf{Q} \text{ is a real-valued}\; N\times N\text{ matrix}\text{;} \label{eqn:Q_conditions_start} \\
        & Q_{i, j} \geqslant 0 \, \text{ for all } i \neq j\,;     \label{eqn:app_backwards_forwards_existence} \\
    & \mathbf{Q}\mathbf{1} = \mathbf{0} \text{ (each row of \(\mathbf{Q}\) sums to $0$). } 
    \label{eqn:app_Q_conditions_end} 
\end{align}
    %&\text{and } Q_{i, j} > 0 \,\text{ if and only if }\,  Q_{j, i} > 0\,; 
    %\label{eqn:app_Q_conditions_end}
Such a transition-rate matrix necessarily has non-positive entries along the diagonal because for any diagonal element, we have \(Q_{i, i} = - \sum_{j\neq i} Q_{i, j}\), which is negative provided there is a non-zero transition rate from state \(i\) to any other state.
The total occupancy of the system can be written as a scalar product of the states with a vector of ones, that is $ \mathbf{1}^\top \mathbf{x}$. 
We can see this total occupancy of the system, \(\mathbf{1}^\top \mathbf{x}\), is conserved because,
\begin{equation}
    \odv{}{t}\left(\mathbf{1}^\top \mathbf{x} \right)
    =\mathbf{1}^\top \odv{}{t} \mathbf{x} 
    = \mathbf{1}^\top \left(\mathbf{Q}^\top\mathbf{x}\right) 
    = \left(\mathbf{1}^\top \mathbf{Q}^\top\right) \mathbf{x}
    = \left(\mathbf{Q1} \right)^\top \mathbf{x} 
    = \mathbf 0\,,
\end{equation}
using Condition~\eqref{eqn:app_Q_conditions_end}.
For this reason, we refer to Condition~\eqref{eqn:app_Q_conditions_end} as the \emph{conservation constraint}.  

Because \(Q_{i, i} = - \sum_{j\neq i} Q_{i, j}\) we have \(\odv{x_i}{t} \geq Q_{i, i} x_i\) and so, if \(x_i(0) > 0\), we have  \(x_i(t) \geq x_i(0) e^{Q_{i, i} t} > 0\). Therefore, provided \(x_i(0) > 0\), the occupancy of each state remains positive for all \(t\).
Hence, as is common throughout the literature \citep{fink_markov_2009}, we may constrain  our initial conditions such that  \(\mathbf{1}^\top \mathbf{x}(0) = 1\). In this way, \(\mathbf{x}(t)\) may be interpreted as a vector of probabilities summing to one. 

We also require that every state in the model is reachable from every other state (though not necessarily directly via a single transition). 
Models satisfying this condition in addition to Conditions~\eqref{eqn:Q_conditions_start}--\eqref{eqn:app_Q_conditions_end}  are said to be \emph{irreducible}. Irreducible Markov models are commonly encountered in mathematical biology, and have a number of desirable properties, as discussed in Section~\ref{sec:app_steady_states}.
Examples of various models satisfying these conditions are shown in Figure~\ref{fig:example_models}.
% Such models are the primary focus of this work, where we show they have desirable properties and prove useful.

\begin{figure}[htbp]
    \centering
    \includegraphics[width=0.8\linewidth]{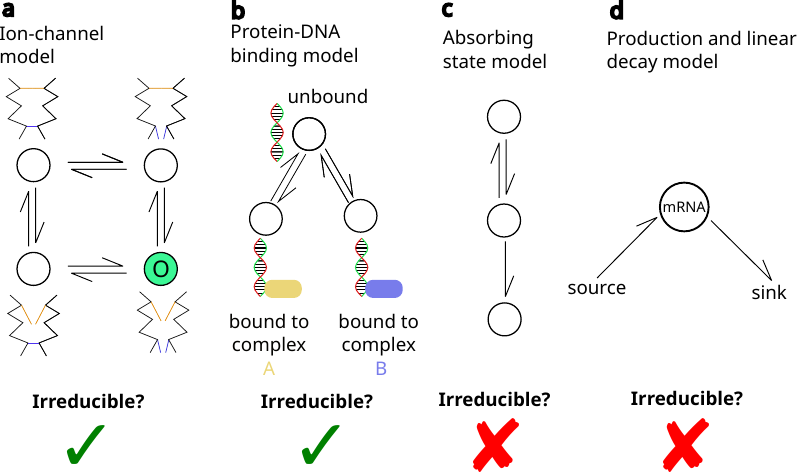}
    \caption{
    Examples of reducible and irreducible Markov models for biological applications. 
    The irreducible models shown in Panels \textbf{a} and \textbf{b} are explored in this paper. 
    The absorbing-state model in Panel \textbf{c} can be seen to be reducible, as there is no path from state \(R\) to state \(I\). 
    Whereas the model in \textbf{d}, representing the production and decay of mRNA,  contains source and sink terms, and hence cannot be irreducible. 
    }
    \label{fig:example_models}
\end{figure}

Note that the above describes models with constant transition rates, but biological models may include transition rates which vary over time.  These transition rates may depend on other variables (such as transmembrane potential in certain ion-channel models). For example, let \(V\) be some external variable, and let \(k(V)\) be a transition rate that is dependent on \(V\). Then if for all such \(k\), we have \(k(V) > 0\) for all \(V\), the existence of a non-zero transition rate between any two states is unaffected by changes to \(V\). In other words, the {model topology} is independent of \(V\).

Key results regarding the existence and uniqueness of steady states (presented in Section~\ref{sec:app_steady_states}) depend only on this model topology, and these properties are therefore independent of external variables like \(V\). 
Both the examples covered in Sections~\ref{sec:example_ion_channel} and \ref{sec:example_auxin} concern irreducible Markov models which include some dependence on time-dependent external variables. Nevertheless, despite concerning only cases where \(\mathbf{Q}\) is constant, the results discussed in Section~\ref{sec:app_steady_states} and \ref{sec:app_reversibility} prove useful for these examples.

We begin in Section~\ref{sec:app_steady_states} by connecting the properties of \(\mathbf{Q}\) with its graphical representation, and show how graph-theoretic properties ensure the existence of globally stable equilibria. Here Proposition~\ref{thm:unique_global_equilibrium} shows that such equilibria are guaranteed to exist provided \(\mathbf{Q}\) is irreducible, and/or equivalent graphical conditions are met.

In Section~\ref{sec:app_reversibility}, we discuss an additional property which is often imposed on Markov models---\emph{microscopic reversibility}. 
We then relate this property to the qualitative behaviour of the related governing equation (Equation~\eqref{eqn:governing_equation}) via the spectrum of the transition-rate matrix, \(\mathbf{Q}\).

In Section~\ref{sec:computational}, we discuss some computational consideration for Markov models in the form of Equation~\eqref{eqn:governing_equation}. 
Here, we derive a method which uses the conservation constraint to reduce the \(N\times N\) ODE-system (Equation~\eqref{eqn:governing_equation}) to an \((N-1)\times(N-1)\) system.
This model simplification can significantly reduce the computation required to simulate a Markov model.

Finally, in Section~\ref{sec:examples}  we demonstrate these approaches through two numerical examples: the first is taken from the study of macroscopic ion-channel kinetics in cardiac electrophysiology \citep{fink_markov_2009}; the second demonstrates how our results may be used to simplify cellular signalling pathways models which contain some subnetwork which satisfies the above properties.
For the first example, we show how the previously introduced results allow more efficient simulation.
%---which is particularly useful for model fitting.
For the second, we use these results to make a quasi steady-state assumption which greatly simplifies simulation and  analysis of the resultant dynamical system. 
We suggest that this material will prove useful more widely for the computational implementation of similar models in any relevant application area.

\section{Steady states of Markov models} \label{sec:app_steady_states}
One particularly important aspect of ODE-based models is their long-term behaviour \citep{hirsh_differential_2013}.  Determining the existence and uniqueness of our models' steady states, is often relied upon for parameter fitting \citep{fink_markov_2009, shuttleworth_evaluating_2025}.
In this section, we show that irreducibility (as described above) guarantees stability and the uniqueness of steady states of Equation~\eqref{eqn:governing_equation}. We begin by introducing some useful graph theory terminology to more precisely describe the model topology.

\subsection{Connectedness and Irreducibility}  \label{sec:connectedness}
The transition-rate matrix, \(\mathbf{Q}\), can be associated with a directed graph or \emph{digraph} \citep{adrian_bondy_graph_2008} having a directed edge (connection) between states \(i\) and \(j\) with weight \(Q_{i, j}\) where \(Q_{i, j} \neq 0\). 
Note that a weighted digraph may also be called a \emph{network} \citep{kelly_reversibility_1981}.

Let \(D\) be a digraph and let \(H \leqslant D\) be a subdigraph, that is, a digraph such that each node in \(H\) is also in \(D\) and each edge in \(H\) is also present in \(D\) \citep{adrian_bondy_graph_2008}. 
Then, \(H \leqslant D\) is said to be \emph{strongly connected} if, for any two nodes, \(u, v \in H\) there is a walk from \(u\) to \(v\). For example the models shown in Figure~\ref{fig:example_models}(\textbf{a}--\textbf{b}) are strongly and weakly connected; whereas the model shown in Figure~\ref{fig:example_models}(\textbf{c}) is weakly connected, but not strongly connected.
In this way, a digraph can be partitioned into \emph{strongly connected components}, each of which is itself strongly connected. 
Equivalently, a digraph has exactly one strongly connected component if and only if it is strongly connected.  For the Markov models discussed in this paper, it is simple to check that a given graph is connected either visually or by using an algorithmic approach such as a depth-first search \citep{adrian_bondy_graph_2008}---this is particularly simple for the Markov models discussed in this paper which have at most \(6\) states.  

A matrix, say \(\mathbf{A}\), is said to be irreducible if there exists no \emph{permutation} matrix \(\mathbf{M}\) (consisting only of \(0\)s and exactly one \(1\) in each row and column)  such that \(\mathbf{MAM}^{-1}\) is upper-triangular \citep{BermanAbraham1979Nmit}.  For Markov models in the form of Equation~\eqref{eqn:governing_equation}, it is known that \(\mathbf{Q}\) is irreducible if and only if the digraph \(D\) is strongly connected \citep{BermanAbraham1979Nmit}.
As we will see in the following section, this link between these properties of the model topology (characterised by \(D\)), and the properties of the transition-rate matrix \(\mathbf{Q}\) prove useful in characterising the behaviour of Markov models. 
Principally, we explain how these graph-theoretic properties are related to the large-time behaviour of Equation~\eqref{eqn:governing_equation}---and in particular, the
\emph{equilibrium point(s)} of Equation~\eqref{eqn:governing_equation}, that is where \(\odv{\mathbf{x}}{t} = \mathbf{Q}^\top \mathbf{x} = 0\). 
Our main focus in this section is the following  result, which closely follows from standard material pertaining to the class of nonnegative matrices \citep{BermanAbraham1979Nmit}.

\begin{proposition} 
\label{thm:unique_global_equilibrium}
Let \(\mathbf{Q}\) be an irreducible matrix satisfying conditions~(\ref{eqn:Q_conditions_start})--(\ref{eqn:app_Q_conditions_end}) and consider the ODE system defined by Equation~\eqref{eqn:governing_equation}, then there is a vector \(\mathbf{x}_\infty \neq \mathbf{0}\) such  that,
\begin{enumerate}[label=(\roman*)]
\item \(\mathbf{Q}^\top \mathbf{x}_\infty = 0\),
\item \(\mathbf{x}_\infty\) is unique up to a scalar,
\item provided \(\mathbf{1}^\top\mathbf{x}(0) = 1\), we have  \(\mathbf{x}(t) \rightarrow \mathbf{x}_\infty \)  as \(t \rightarrow \infty\).
\end{enumerate}
\end{proposition}

In other words, if the  digraph of our Markov model is {strongly connected},  there exists a unique stable equilibrium point. Note that in the case where each reaction in the model, \(i \rightarrow j\) with \({Q}_{i, j} > 0\), has a corresponding reverse reaction (that is, \(Q_{j, i} > 0\)), we need only consider the connectedness of the (undirected) graph representation of the model. 

We first focus on part \textit{(i)}, the existence of an equilibrium point for Equation~\eqref{eqn:governing_equation}. 
The proofs of parts \textit{(ii)} and \textit{(iii)} are deferred to Section~\ref{sec:perron}, following a discussion of the necessary theory. 
We begin by considering the eigenvalues of \(\mathbf{Q}^\top\). 
Firstly, we see that the condition that \(\mathbf{Q} \mathbf{1} = \mathbf{0}\)
means precisely that \(0\) is an eigenvalue of \(\mathbf{Q}\) with \(\mathbf{1}\) being the corresponding eigenvector.
Hence, we see that,
\[
\mathbf{0}^\top = (\mathbf{Q}\mathbf{1})^\top 
= \mathbf{1}^\top \mathbf{Q}^\top,
\]
and so \(\mathbf{1}^\top\) is a left-eigenvector of \(\mathbf{Q}^\top\) with eigenvalue \(0\). 
The characteristic equation is not changed by transposition, so  \(\textrm{det}(\mathbf{Q} - \lambda \mathbf{I}) = \textrm{det}(\mathbf{Q}^\top - \lambda \mathbf{I})\) for any \(\lambda\). Hence, setting \(\lambda=0\), we see that there is a corresponding right eigenvector, \(\mathbf{v}\) with \(\mathbf{Q}^\top \mathbf{v} = 0\), which means that \(\mathbf{v}\) is an equilibrium point of the governing equation (Equation~\eqref{eqn:governing_equation}), hence,
\begin{equation}
  \odv{\mathbf{x}}{t} =  \mathbf{Q}^\top \mathbf{v} = \mathbf{0}. \label{eqn:eqm_point}
\end{equation} 

%Note, however, that this equilibrium point is not unique. 
Note that for any \(\alpha \in \mathbb{R}\), we also have \(\mathbf{Q}^\top (\alpha \mathbf{v}) = \alpha \mathbf{Q}^\top\mathbf{v} = 0\), but
the only equilibrium point on this line satisfying the constraint \(\mathbf{1}^\top \mathbf{v} = 1\) is\(\dfrac{\mathbf{v}}{\mathbf{1}^\top\mathbf{v}}\).  Hence, we have proved part \textit{(i)} of Proposition~\ref{thm:unique_global_equilibrium}.
In Section~\ref{sec:perron}, we provide proofs of parts \textit{(ii)} and \textit{(iii)} which concern the uniqueness and stability of  the equilibrium \(\mathbf{x}_\infty\)

The following result concerns the properties of the non-zero eigenvalues. 
These results allow us to assess the stability of the equilibria described above as when  an eigenvalue of \(\mathbf{Q}^\top\), say \(\lambda\) has  \(\textrm{Re}(\lambda) >0\), the system is unstable \citep{hirsh_differential_2013}.
Later, in Section~\ref{sec:perron}, we use this result to prove parts \textit{(ii)} and \textit{(iii)} of Proposition~\ref{thm:unique_global_equilibrium}.
\begin{proposition} \label{prop:eigenvalues_lessthan_zero}
    Let \(\mathbf{Q}\) be a matrix satisfying conditions (\ref{eqn:Q_conditions_start}--\ref{eqn:app_Q_conditions_end}). 
    Then, for any eigenvalue of \(\mathbf{Q}\), \(\lambda\), we have \(\Re(\lambda) \leq 0\).
\end{proposition}

\textit{Proof of Proposition~\ref{prop:eigenvalues_lessthan_zero}:} by Conditions  (\ref{eqn:Q_conditions_start})--(\ref{eqn:app_Q_conditions_end}), we have
\begin{equation}
    |Q_{i, i}|  = \left| \sum_{j\neq i} Q_{i, j} \right| = \sum_{j\neq i}\left|Q_{i,j}\right|,
\end{equation} 
because \(Q_{i, j} \geqslant 0\) for all \(i, j\) with \(i \neq j\). 
In other words, \(\mathbf{Q}\) is \emph{(weakly) diagonally dominant}. Because the diagonal elements are negative, weak diagonal dominance implies that the eigenvalues of \(\mathbf{Q}\) have non-positive real part \citep[Chapter~6.1]{Horn_Johnson_2012}.
Then, because the eigenvalues of \(\mathbf{Q}\) and \(\mathbf{Q}^\top\) are identical, all non-zero eigenvalues of  \(\mathbf{Q^\top}\) have non-positive real part---that is, \(\Re(\lambda) \leq 0\) for each eigenvalue, \(\lambda\). 

Next, we demonstrate these eigenvalue properties for two example model topologies, the first, is a strongly connected model where \(\mathbf{Q}\) is irreducible, and the second, an example of a Markov model with a reducible transition-rate matrix, for which there is no unique equilibrium.

\subsubsection{Example: a disconnected Markov model}
Consider the four-state Markov model, \begin{equation}
   \mathbf Q = \begin{pmatrix}
    -k_1 & k_1 & 0 & 0 \\
    k_2 & -k_2 & 0 & 0 \\
    0&  0 & - k_3 & k_3 \\
    0 & 0 & k_4 & -k_4
    \end{pmatrix}.
    \end{equation}
We can see from the corresponding graph (Figure~\ref{fig:example_graphs}(\textbf{b})) that this Markov model is disconnected. 
Alternatively, we can see that \(\mathbf{Q}\) is reducible because it is already in block upper-triangular form. 
In computing the eigenvalues of \(\mathbf{Q}\), we see that it has two eigenvectors corresponding to the zero eigenvalue. 
The two connected components behave independently (as shown in Figure~\ref{fig:example_graphs}), and so we have steady states of the form \((v_1, v_2, 0, 0)^\top\) and \((0, 0, v_3, v_4)^\top\) where \(v_1, v_2, v_3, v_4 \neq 0\).

Such models may be treated as two separate ODE systems, representing the first two and final two states, respectively, with each individually satisfying \(\mathbf{1}^\top \mathbf{x} = 1\).
Proposition~\ref{thm:unique_global_equilibrium} can then be applied to each component individually to see that the combined system has a unique equilibrium
Such an approach is often taken in the computation of cardiac ion-channel models
\citep{keener2009mathematical, beattie_sinusoidal_2018, shuttleworth_evaluating_2025}.
So-called Hodgkin-Huxley ion-channel models may take this form, where multiple independent open-closed components (referred to as subunits) are combined to give a single output (the proportion of open channels) \citep{keener2009mathematical, fink_markov_2009}.

For typically-sized models, it is easy to visually check that a given transition-rate matrix is strongly connected.
Alternatively, it is possible check connectedness programmatically by using an algorithm such as Depth-First Search \citep{gibbons_algorithmic_1985}.

\subsubsection{Example: a connected Markov model}
Consider the three-state Markov model shown as a weighted digraph in Figure~\ref{fig:example_graphs}(\textbf{a}). 

\begin{figure}[htbp]
    \centering
    \includegraphics{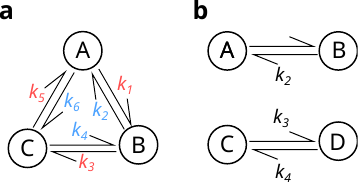}
    \caption{\textbf{a}: a three-state Markov model. 
    This graph is strongly connected as each state is accessible from every other state, so the corresponding transition rate matrix must be irreducible, and the corresponding system has a unique equilibrium point (under the constraint \(\mathbf{1}^\top\mathbf{x}(0) = 1\)). 
    The transition rates oriented in the clockwise direction (red) have odd indices (\(k_1, k_3, k_5\)). Whereas the transition rates pointing in the anti-clockwise direction (blue) have even indices (\(k_2, k_4, k_6\)). \textbf{b}: a four-state Markov model with a disconnected graph. Since the graph is disconnected, the corresponding transition-rate matrix is reducible. 
    The two connected components (A and B, and C and D) may be seen as wholly independent processes, each with a unique stable equilibrium.
    }
    \label{fig:example_graphs}
\end{figure}

The transition rate matrix for $\boldsymbol{x}=(A, B, C)^\top$ is,
\begin{equation}
   \mathbf Q = \begin{pmatrix}
    -k_1-k_6 & k_1 & k_6 \\
    k_2 & -k_2 - k_3 & k_3 \\
    k_5 & k_4 & - k_4 - k_5
    \end{pmatrix}. \label{eqn:3-state-example1}
\end{equation} 

We can see that \(\mathbf{Q}\) is irreducible from Figure~\ref{fig:example_graphs} as the corresponding digraph is strongly connected. 
As discussed above, this implies there is exactly one eigenvector of \(\mathbf{Q}^\top\) corresponding to a vanishing eigenvalue. 

We may compute the eigenvalues of Equation~\eqref{eqn:3-state-example1} directly.
The characteristic polynomial of \(\mathbf{Q}^\top\) is,
\begin{align}
    &\textrm{det}(\mathbf{Q}^\top - \lambda \mathbf{I})\\
    &= \lambda \big(-k_1 k_3 - k_1 k_4 - k_2 k_4 - k_1 k_5 - k_2 k_5 - k_3 k_5 - k_2 k_6 - k_3 k_6 - k_4 k_6 \nonumber \\& \hspace{4em}- k_1 \lambda - k_2 \lambda - k_3 \lambda - k_4 \lambda - k_5 \lambda - k_6 \lambda - \lambda^2 \big). \nonumber
\end{align}
Hence, \(0\) is an eigenvalue of multiplicity \(1\), meaning that there is an eigenvector \(\mathbf{v}\) satisfying \(\mathbf{Q}^\top \mathbf{v} = \mathbf 0\) which is unique up to a scalar. That is, any \(\mathbf{u}\) with \(\mathbf{Q}^\top\mathbf{u} = 0\), may be expressed as a scalar multiple of \(\mathbf{v}\), say, \(\mathbf{u} = \alpha \mathbf{v}\) for some \(\alpha\in\mathbb{R}\).
Therefore, if we include the constraint that \(\mathbf{1}^\top\mathbf{x} = 1\), the governing equation has exactly one equilibrium point: \(\dfrac{\mathbf{v}}{\mathbf{1}^\top\mathbf{v}}\).

\subsubsection{A Markov model with complex eigenvalues}
We conclude by demonstrating the results introduced in this section by computing the eigenvalues of \(\mathbf{Q}^\top\) for a three-state model as shown in Figure~\ref{fig:example_graphs}. 
We set transition rates as, 
\begin{equation}
        (k_1,\, k_2,\, k_3,\, k_4,\, k_5,\, k_6) =     
     (1,\, 2,\, 3,\, 4,\, 5,\, 10).
\end{equation} 
The resulting transition rate matrix is,
\begin{equation}
    \mathbf{Q} = \begin{pmatrix}
        -11 & 1 & 10 \\ 
        2 & -5 & 3 \\
        5 & 4 & -9
    \end{pmatrix}.
\end{equation}
To one decimal place, the right-eigenpairs of \(\mathbf{Q}^\top\) are,
\begin{align*}
    &\lambda_1 = -17.1,  & \mathbf{v}_1 &= (-0.7, -0.3, 1)^\top , \\
    &\lambda_2  =-7.9,  & \mathbf v_2 &= (0.6, -1.6, 1)^\top , \\ \text{and }\,
    &\lambda_3 = 0, & \mathbf v_3 &= (0.6, 0.9, 1)^\top.
\end{align*}

However, if we choose,
\begin{equation}
(
    k_1,
    k_2,
    k_3,
    k_4, 
    k_5, 
    k_6 
)  =
(
   1,
    100,
    1,
    100, 
    1,
    100
).
\end{equation} the right-eigenpairs of \(\mathbf Q^\top\) are (to one decimal place), 
%https://www.wolframalpha.com/input?i=%7B%7B-k1-k6%2C+k1%2C+k6%7D%2C+%7Bk2%2C+-k2-k3%2C+k3%7D%2C+%7Bk5%2C+k4%2C+-k4+-k5%7D%7D+with+%5Bk1%2C+k2%2C+k3%2C+k4%2C+k5%2C+k6%5D+%3D+%5B1%2C+100%2C+1%2C+100%2C+1%2C+100%5D
\begin{align*}
    &\lambda_1 \approx -151.5 + 85.7i , &v_1 &\approx (-0.5 + 0.9i, -0.5 - 0.9i, 1.0), \\
     &\lambda_2 \approx -151.5 - 85.7i, &v_2 &\approx(-0.5 - 0.9i, -0.5 + 0.9i,  1.0),\\
     &\lambda_3 \approx 0.0, &v_3 &\approx (1.0, 1.0, 1.0)\,.
\end{align*}

This example demonstrates that the eigenvalues of \(\mathbf{Q}\) may be non-real for certain choices of \(\mathbf{Q}\), even though the solution itself is real-valued. 
This fact should be considered when a computational eigendecomposition method \citep{moler2003nineteen} is used to simulate such a model (that is, to solve Equation~\eqref{eqn:governing_equation}). 
However, the computational implementation of other matrix-exponential solution methods does not require complex variables in computer code (since solutions and the transition-rate matrix itself contain only real values).
Examples of such methods include Pad\'e approximation \cite{moler2003nineteen}, as used in \cite{teed_computationally_2016-1}, and the use of numerical ODE solvers, as presented in Section~\ref{sec:examples}.
Also, as discussed in Section~\ref{sec:app_reversibility},  non-real eigenvalues do not arise in Markov models which satisfy microscopic reversibility---an additional criterion often imposed on the transition-rate matrix, \(\mathbf{Q}\).

\subsection{The Perron-Frobenius Theorem for Markov models} \label{sec:perron}

The Perron-Frobenius theorem is a collection of powerful results arising from the study of strictly positive matrices, that is, those matrices \(\mathbf{A}\) with the property, \(A_{i,j} > 0\) for all \(i, j\) \citep{BermanAbraham1979Nmit, meyer_matrix_2000}.
However, there are certain extensions of Perron-Frobenius theory to certain classes of nonnegative matrices, such as the theorem presented below.
The theorem presented below is a simplified version of the result presented by \citet{meyer_matrix_2000} concerning nonnegative, irreducible matrices.
However, stronger results hold for strictly positive matrices
\citep{meyer_matrix_2000, seneta06, BermanAbraham1979Nmit}. 

Note that a transition-rate matrix, \(\mathbf{Q}\) satisfies Conditions~\eqref{eqn:Q_conditions_start}--\eqref{eqn:app_Q_conditions_end}  necessarily contains negative elements, and so these Perron-Frobenius results cannot be applied directly to \(\mathbf{Q}\). 
Instead, we prove Proposition~\ref{thm:unique_global_equilibrium} by applying this result to a closely-related nonnegative matrix.

\begin{theorem}[Perron-Frobenius Theorem for nonnegative irreducible matrices]\citep[Chapter~8]{meyer_matrix_2000}
\label{thm:perron_frobenius}
Let \(\mathbf{M} = (M_{i,j})\) be an irreducible \(N \times N\) matrix with nonnegative entries. Then there exists an eigenpair of \(\mathbf{M}\), \((r,\, \mathbf{v}^+)\) such that:
\begin{itemize}
    \item  \(r \geqslant |\lambda|\) for all eigenvalues \(\lambda\);
    \item the components of \(\mathbf{v}^+\), \(v_i^+\), satisfy \(v_i^+ > 0\) for all \(i\);
    \item and, excluding scalar multiples, \(\mathbf{v}^+\) is the only eigenvector of \(\mathbf{M}\) satisfying \(\mathbf{M}\mathbf{v}^+ = r\,\).
    \end{itemize}
\end{theorem}
Such an eigenvalue, \(r\), is referred to as the \emph{leading eigenvalue} \citep{BermanAbraham1979Nmit}.

% Below argument is briefly covered on page 153 of B&P nonnegative matrices, talking about semi convergence and "M" matrices
% -Q fits the definition of an M matrix (no positive off-diagonal elements and nonnegative eigenvalues) - easy to see because of diagonal dominance
% Maybe confusing to use M here, because it isn't an 'M matrix'
% Maybe there's a couple too many logic leaps in here to cite it directly
Whilst transition-rate matrices of the type introduced in Section~\ref{sec:introduction} do not satisfy this condition, we can see that,
\(\mathbf{A} = \mathbf{Q} + k\mathbf{I}\) does for \(k = \max_{i}\{-Q_{i, i}\}\). 
Moreover, noting that the graphs associated with \(\mathbf{A}\) and \(\mathbf{Q}\) are identical, we see that \(\mathbf{A}\) is also irreducible, and so we can apply Theorem~\ref{thm:perron_frobenius} to \(\mathbf{A}\).
%See Corollary 2.28 in B&P

Consider the characteristic polynomial of \(\mathbf{A}\),
\begin{align}
    \det(\mathbf{A} - \lambda \mathbf{I}) =&\det(\mathbf{Q} + k\mathbf{I} - \lambda \mathbf{I})\,,
\end{align}
from which we can see a simple relationship between the eigenvalues of \(\mathbf{A}\) and \(\mathbf{Q}\) given by the one-to-one mapping between eigenvalues:
if \(\lambda\) is an eigenvalue of \(\mathbf{A}\), then \(\lambda -k\) is an eigenvalue of \(\mathbf{Q}^\top\), and the corresponding eigenvectors are equal. 
 Note also that for any \(i, j\) with \(i \neq j\), we have that \(A_{i, j} > 0 \) if and only if \(Q_{i, j} > 0\). Hence, if \(\mathbf{Q}\) is irreducible, \(\mathbf{A}\) is too.

We know from Proposition~\ref{prop:eigenvalues_lessthan_zero} that all eigenvalues of \(\mathbf{Q}^\top\) satisfy \(\Re(\mu) \leqslant 0 \), and that \(0\) is one of these eigenvalues. Hence, \(k\) must be the leading eigenvalue of \(\mathbf{A}\). 
Moreover, applying the Perron-Frobenius theorem to \(\mathbf{A}\) tells us that for each eigenvalue of \(\mathbf{A}\), \(\lambda\),  satisfies \begin{equation}
    |\lambda | \leqslant k,
\end{equation} 
and so, each eigenvalue of \(\mathbf{Q}^\top\), say \(\mu\), satisfies, 
\begin{equation}
    |\mu + k| \leqslant k\,.
\end{equation}
Therefore, we have that \(\textrm{Re}(\mu) \leqslant 0 \) for all eigenvalues \(\mu\), and \(\textrm{Re}(\mu) = 0\) only when \(\mu\) corresponds to the leading eigenvalue of \(\mathbf{A}\).

The Perron-Frobenius Theorem also tells us that the leading eigenvalue, \(r\), corresponds to some  eigenvector \(\mathbf{v}^+\) with \(v^+_i > 0\) for all \(i\).
Since \(0\) is the eigenvalue of \(\mathbf{Q}^\top\) with maximal real part, the eigenpair of \(\mathbf{A}\), \((r,\, \mathbf{v}^+)\) must correspond to an eigenpair of \(\mathbf{Q}^\top\), \((0,\, \mathbf{v}^+)\).
Also, from Theorem~\ref{thm:perron_frobenius}, we have that \(\mathbf{v}^+\) is unique up to scalar multiplication.
This proves Proposition~\ref{thm:unique_global_equilibrium}. 

Note however, these results provide no guarantee that eigenvalues (except the leading eigenvalue) are real-valued. 
As demonstrated above, some Markov models, with particular choices of transition rates, have non-real eigenvalues, leading to oscillatory solutions. 
The realness of these eigenvalues may be guaranteed under certain conditions that we discuss next.

\section{Microscopic Reversibility} 
\label{sec:app_reversibility}
In addition to Conditions (\ref{eqn:Q_conditions_start}--\ref{eqn:app_Q_conditions_end}),  \emph{microscopic reversibility} is another property that is often imposed on Markov models and concerns the balance of fluxes between individual states in the model. 
Consider a simple reaction \(A\rightleftharpoons B\), then microscopic reversibility is satisfied if, at equilibrium, the flux from \(A\) to \(B\) is equal to the flux of the reverse reaction \(B\) to \(A\). 
A model is said to satisfy microscopic reversibility if this is true for all reactions in the model. 
A more precise definition for Markov models is given below.

Microscopic reversibility is varyingly referred to as ``detailed balance'', ``thermodynamic equilibrium'' or sometimes simply ``reversibility'' \citep{colquhoun_how_2004}. 
Microscopic reversibility is seen as a natural, macroscopic consequence of the laws of thermodynamics when the model is assumed to be a closed system at equilibrium \citep{ederer2007thermodynamically, crooks2011thermodynamic}.
Although, microscopic reversibility has not always been enforced in ion-channel models (for example in \citet{li_improving_2017}).
Additionally, microscopic reversibility often leads to a reduction in the number of model parameters---a particularly attractive property when fitting models \citep{bates_nonlinear_1988, colquhoun_how_2004, fink_markov_2009}. 

For brevity, we proceed using the term ``reversibility'', and adopt the following definition.
\begin{definition}
Consider a Markov model which satisfies Conditions  (\ref{eqn:Q_conditions_start}--\ref{eqn:app_Q_conditions_end}) with transition-rate matrix, \(\mathbf{Q}\) such that \(\mathbf Q^\top\) is irreducible. 
Then by Proposition~\ref{thm:unique_global_equilibrium}, a  unique equilibrium point, \(\mathbf x_\infty=\left(x_\infty^{(i)}\right)_{i=1}^{N}\), exists. 
We say that this Markov model is \emph{reversible} if it satisfies,
\begin{equation}
    Q_{i, j} {x}_\infty^{(i)} = Q_{j, i} {x}_\infty^{(j)}, \label{eqn:detailed_balance}
\end{equation} for all \(i\) and \(j\). That is, at equilibrium, the flux from state \(i\) to state \(j\) is equal to the flux from state \(j\) to state \(i\).
This is also commonly known as a detailed balance equation \cite{keener2009mathematical}.
\end{definition}
Note that this definition of reversibility is wholly concerned with properties of \(\mathbf{Q}\) because \(\mathbf{x}_\infty\) is precisely the right-eigenvector of \(\mathbf{Q}^\top\) with \(\mathbf{Q}^\top\mathbf{x}_\infty = 0\), that is, the null eigenvector, which is unique up to a scalar as discussed in Section~\ref{sec:app_steady_states}. 
% We may extend this definition to the transition-rate matrix and call \(\mathbf{Q}\) \emph{reversible} if and only if the corresponding Markov model is reversible. 

Note that this definition does not strictly agree with the definition found in \citet{kelly_reversibility_1981}, where to be reversible, a CTMC must be \emph{stationary}---analogous to a Markov model initially being at equilibrium, that is \(\mathbf{x}(0) = \mathbf{x}_\infty\). In contrast to this, the above definition merely concerns the properties of \(\mathbf{Q}\), not the model's initial conditions.
The following result shows that there may be qualitative differences in the behaviour of reversible Markov models and non-reversible Markov models.

\begin{proposition} \label{prop:reversible_real_eigenvalues}
Let \(\mathbf Q\) be an irreducible matrix satisfying conditions  (\ref{eqn:Q_conditions_start})-(\ref{eqn:app_Q_conditions_end}). Additionally, suppose \(\mathbf{Q}\) is reversible. 
Then, the eigenvalues and eigenvectors of \(\mathbf{Q}^\top\) are real.    
\end{proposition}
We prove that above result by assuming that our model is reversible. Then, we may construct the diagonal matrix \(\mathbf{D}\) with \(D_{i, i} = \sqrt{x_\infty^{(i)}}\) for all \(i\). 
\(\mathbf{D}\) is then a diagonal matrix, and from Proposition~\ref{thm:unique_global_equilibrium} we see that the irreducibility of \(\mathbf{Q}\) guarantees that \(\mathbf{D}\) has no zeroes on its diagonal. 
Hence, \(\mathbf{D}\) is invertible. 
Then, we have,
\begin{align}
(\mathbf{D} \mathbf{Q} \mathbf{D}^{-1})_{j, i} 
&= \sqrt{{x}_\infty^{(j)}}  Q_{i, j} \dfrac{x_\infty^{(i)}}{x_\infty^{(j)}} \dfrac{1}{\sqrt{{x}_\infty^{(i)}}} , \; \; \; \text{(using Equation~\eqref{eqn:detailed_balance})} \nonumber \\
&= \sqrt{{x}_\infty^{(i)}}  Q_{i, j} \dfrac{1}{\sqrt{{x}_\infty^{(j)}}},\nonumber \\
&= (\mathbf D \mathbf Q \mathbf D^{-1})_{i, j}.
\end{align}
Now, the characteristic polynomial of \(\mathbf{Q}\)  is identical to that of  \(\mathbf{D}\mathbf{Q}\mathbf{D}^{-1}\) meaning these matrices have identical eigenvalues. And because \(\mathbf{D}\mathbf{Q}\mathbf{D}^{-1}\) is symmetric, these eigenvalues are real \citep{meyer_matrix_2000}.
Therefore, whenever microscopic reversibility holds, \(\mathbf{Q}\), and hence \(\mathbf{Q}^\top\), have only real eigenvalues.
Consequently, eigendecomposition matrix-exponential methods \citep{moler2003nineteen} can be used to solve the governing equation and simulate the model.
Without microscopic reversibility, there is no such guarantee, and numerical implementations should account for non-real eigenvalues and eigenvectors.

This result provides some insight into the implications  of reversibility on the dynamics of Markov models. 
We have seen above that non-real eigenvalues are, in fact, possible for non-reversible models.
These eigenvalues, with \(\Im(\lambda) \neq 0\), describe oscillatory behaviour (similar to a damped oscillator when \(\textrm{Re}(\lambda) \leqslant 0 \)). 
Here, we have shown that such eigenvalues, and hence oscillatory behaviours, are not possible when \(\mathbf{Q}\) is reversible. 

Returning to the three-state Markov model discussed in the previous section (see Equation~\eqref{eqn:3-state-example1}), we can see that such a Markov model is not reversible in general because the steady state is,
\begin{equation}
    \mathbf{x}_\infty = \alpha \begin{pmatrix}
        \dfrac{k_2k_4 + k_2k_5 + k_3k_5}{k_1k_3 + k_2k_6 + k_3k_6} \\ \\
        \dfrac{k_1k_4 + k_1k_5 + k_4k_6}{k_1k_3 + k_2k_6 + k_3k_6}\\ \\
        1
    \end{pmatrix},
\end{equation} 
for some scaling factor, \(\alpha\).  
Now considering the flux from state \(1\) to state \(2\) we have,
\begin{align}
x_\infty^{(1)}Q_{1,2} 
    &= \alpha k_1  \dfrac{k_2k_4 + k_2k_5 + k_3k_5}{k_1k_3 + k_2k_6 + k_3k_6},
\end{align}  
whilst, the flux in the opposite direction is,
\begin{align}
    x_\infty^{(2)}Q_{2,1} 
    =& \alpha k_2 \dfrac{k_1k_4 + k_1k_5 + k_4k_6}{k_1k_3 + k_2k_6 + k_3k_6}.
\end{align} 
Hence, this Markov model is reversible if and only if,
\begin{equation}
k_1k_3k_5 = k_2k_4k_6. \label{eqn:proto_kolmogorov_3-state}
\end{equation} 
As we saw previously, when our model is not reversible, or equivalently, \(k_1k_3k_5 \neq k_2k_4k_6\), the transition-rate matrix, \(\mathbf{Q}\), may have complex eigenvalues. 
In the following section, we see that a similar condition may be used, in general, to check the reversibility of any  Markov model.

\subsection{Kolmogorov's Criterion} \label{sec:kolmogorov}
It is possible to ensure reversibility by checking that the detailed-balance equation (Equation~\eqref{eqn:detailed_balance}) holds. 
However, this requires us to compute the steady state. 
Kolmogorov's criterion provides an alternative method which is far more convenient. 

\begin{theorem} [Kolmogorov's Criterion \citep{kelly_reversibility_1981}] \label{thm:kolmogorov_reversibility}
    A Markov model with transition-rate matrix \(\mathbf{Q}\), with corresponding graph, \(G\), is reversible if and only if, for every cycle in \(G\), the products of transition rates in each direction are equal. 
    In other words, Equation~\eqref{eqn:detailed_balance} holds if and only if,
    \begin{equation}
        \prod_{i=1}^{C_n} Q_{s_{i-1}, s_{i}} = \prod_{i=1}^{C_n} Q_{s_{i}, s_{i-1}} \label{eqn:app_kolmogorov_condition}
    \end{equation} 
    for every cycle  \(C = (s_0, \ldots s_{C_{n}})\) of \(G\), where \(s_{C_{n}} = s_0\).
\end{theorem}
For example, we saw above in the three-state model that reversibility was equivalent to the condition 
\(k_1 k_3 k_5 = k_2 k_4 k_6\). This is an equality between the product of transition rates around the only cycle in each direction (as illustrated in Figure~\ref{fig:example_graphs}(\textbf{a})).

This result is proved in \citet{kelly_reversibility_1981} for continuous time Markov chains (CTMCs), but an almost identical proof can be applied to ODEs of the form of Equation~\eqref{eqn:governing_equation} satisfying Conditions~\eqref{eqn:Q_conditions_start}--\eqref{eqn:app_Q_conditions_end}. Using this result, we may ensure that a given Markov model is reversible, by simply drawing the diagram of our model (such as those in Figure~\ref{fig:example_models}), and checking whether the product of transition rates around each cycle is equal in both directions. When this condition is satisfied, the Markov model is reversible.
However, as explained by \citet{colquhoun_how_2004}, we need only check that each of a subset of cycles satisfies Equation~\eqref{eqn:app_kolmogorov_condition}---that being the \emph{cycle basis} of \(G\) \citep{yang2006imposing}. 

\section{Simplifying the governing equation}\label{sec:computational}
We conclude by explaining how our Markov models' conservation constraint (Equation~\eqref{eqn:app_Q_conditions_end}) can be used to simplify the governing equation (Equation~\eqref{eqn:governing_equation}). 
As discussed in \citet{fink_markov_2009}, this operation effectively simplifies the ODE system, allowing the governing equation to be written in terms of \(N-1\) state variables instead of \(N\). In Section~\ref{sec:example_ion_channel}, we provide an example where this simplification of the ODE system becomes significantly more efficient to solve. 

Such a simplification of the system of ODEs is commonly applied to Markov models of ion-channel currents \citep{fink_markov_2009, beattie_sinusoidal_2018}. 
\citet{colquhoun1995q} provide a description of this procedure, which may be applied to any Markov model (that is, a system of ODEs satisfying Equations~(\ref{eqn:Q_conditions_start}--\ref{eqn:app_Q_conditions_end})). 
In this section, we provide an explicit derivation of their method, showing that there are many choices of such simplified systems.

First, to derive the scheme suggested by \citet{colquhoun1995q}, we seek a reduced system of ODEs,
\begin{equation}
    \dfrac{\textrm{d}\tilde{\mathbf x}}{\textrm{d}t} = \mathbf{A} \tilde{\mathbf{x}} + \mathbf{b},
\end{equation} 
where \(\tilde{\mathbf{x}}\) is the top \(N - 1\) components of \(\mathbf{x}\), \(\mathbf{A}\) is a \((N-1) \times (N-1)\) matrix, and \(\mathbf{b}\) is a column vector of size \((N-1)\).  
Then, by setting \(x_n = 1 - \sum_{j=1}^{N-1}x_j\), our constraint, \(\mathbf{1}^\top \mathbf{x}=1\) is always satisfied. 

Let \(\mathbf{a}\) and \(\mathbf{b}\) be \((N-1)\) column vectors, and \(c\) be a scalar such that,
\begin{equation}
\mathbf{Q}^\top(V) = 
\begin{pmatrix}
    \tilde{\mathbf{Q}}^\top & \mathbf b \\ \mathbf{a}^\top & c
\end{pmatrix},
\end{equation}
where \(\tilde{\mathbf{Q}}\) is an \((N-1)\times(N-1)\) matrix containing the  elements in the upper-left rows and columns of \(\mathbf{Q}\). 
Then, we multiply both sides of Equation~\eqref{eqn:governing_equation} by the \((N-1) \times N\) matrix,
\begin{equation}
    \begin{pmatrix}
        \mathbf{I_{N-1}} & 0
    \end{pmatrix} =\begin{pmatrix}
    1 & 0 & 0 & \ldots & 0 & 0 \\
    0 & 1 & 0 & \ldots & 0 & 0\\
    0 & 0 & 1 & \ldots & 0 & 0\\
    \vdots & \vdots & & \ddots & \vdots & \vdots \\
    0 & 0 & 0 & \ldots & 1 & 0\\
    \end{pmatrix},
\end{equation} 
which is simply an identity matrix with a column of zeros appended to its right side. Then, we have,
\begin{align}
      \begin{pmatrix}
        \mathbf{I_{N-1}} & \mathbf{0}
    \end{pmatrix} 
    \odv{}{t} \mathbf{x}   &= \left(\begin{array}{l  r}
      \tilde{\mathbf{Q}}^\top \;\;& \mathbf{I}_{N-1} \mathbf{b}
      \end{array} 
      \right) \mathbf{x}, \nonumber \\
       &=
  \begin{pmatrix} \tilde{\mathbf{Q}^\top} \tilde{\mathbf{x}}& \\
  0&
  \end{pmatrix}
      +  x_N\begin{pmatrix} \mathbf{b}\\ 
    0 \end{pmatrix}
    .
\end{align} 
Hence, substituting \(x_N = 1 - \tilde{\mathbf{x}}^\top\mathbf{1}\) and considering the first \(N - 1\) components,
\begin{align}
      \odv{}{t}
      \tilde{\mathbf{x}} =
      \left(\tilde{\mathbf Q}^\top - \mathbf{b}\mathbf{1}^\top \right) \tilde{\mathbf{x}} + \mathbf{b}. \label{eqn:app_simplified_ode1}
\end{align}
After solving this system, it is simple to use the relation \(x_N = 1 - \mathbf{1}^\top\tilde{\mathbf{x}}\), to recover \(\mathbf{x}\) from \(\tilde{\mathbf{x}}\).

However, Equation~\eqref{eqn:app_simplified_ode1} is but one of many ways in which we use the conservation constraint to simplify  Equation~\eqref{eqn:governing_equation}. 
Suppose \(\mathbf{T}\) is a non-singular matrix such that each element of the bottom row is \(1\). 
Then, we have the block matrix,
\begin{equation}
    \mathbf{T} = \begin{pmatrix} \mathbf{U} \\ \mathbf{1}^\top\end{pmatrix}.
\end{equation}

Now, we let \(\mathbf{s} = \mathbf{T}\mathbf{x}\) be the state vector of our new, transformed system of ODEs with,
\begin{align}
    \odv{\mathbf{s}}{t} &= \odv{\mathbf{Tx}}{t} \nonumber \\
                        &= \mathbf{T}\odv{\mathbf{x}}{t} \nonumber  \\
                        &= \mathbf{T}\mathbf{Q}^\top \mathbf{x} \nonumber  \\ 
                        &= \mathbf{T}\mathbf{Q}^\top\mathbf{T}^{-1}\mathbf{s},
                        \label{eqn:governing_equation_s}
\end{align}
which is possible because we chose \(\mathbf{T}\) to be non-singular.
If we let 
\(\mathbf{W} = \mathbf{T}\mathbf{Q}^\top\mathbf{T}^{-1}\)
, we may write,
\begin{equation}
    \odv{\mathbf{s}}{t} = \mathbf{W}\mathbf{s}.
\end{equation}
Since each element of the bottom row of \(\mathbf{T}\) is \(1\), the bottom component of \(\mathbf{s}\) is \(s_N = \mathbf{1}^\top \mathbf{x} = 1\), and so, \(s_N\) is constant. 
Hence, we may simplify Equation~\eqref{eqn:governing_equation_s}, by writing,
\begin{equation}
    \odv{\tilde{\mathbf{s}}}{t} = \mathbf{C}\tilde{\mathbf{s}} + \mathbf{d},
\end{equation}
where \(\tilde{\mathbf{s}}\) is the first \(N-1\) rows of \(\mathbf{s}\), \(\mathbf{C}\) is the first \(N-1\) rows and columns of \(\mathbf{W}\), 
and \(\mathbf{d}\) is first \(N-1\) entries of the rightmost column of \(\mathbf{W}\).
Now, consider the bottom row of \(\mathbf{W}\), denoted \(\mathbf{W}_N\).
We have, \(\mathbf{W}_N = (\mathbf{T}\mathbf{Q}^\top\mathbf{T}^{-1})_N\),
and \((\mathbf{T}\mathbf{Q}^\top)_N = \mathbf{1}^\top\mathbf{Q}^\top = \mathbf{0}\) (using Condition \eqref{eqn:app_Q_conditions_end}), and so,
\(\mathbf{W}_N = \mathbf{0}^\top \mathbf{T}^{-1}= \mathbf{0}^\top\). Hence, \(\mathbf{W}\) is block upper-triangular with determinant,
\begin{equation}
    \textrm{det}(\mathbf{W} - \lambda \mathbf{I}) = -\lambda \textrm{det}(\mathbf{C} - \lambda \mathbf I)\,.
\end{equation}
Also, Proposition~\eqref{thm:unique_global_equilibrium}(ii) guarantees that the null eigenvalue is simple (corresponding to a unique equilibrium) meaning that \(\textrm{rank}(\mathbf{W}) = N - 1\) and so,
\(\mathbf{C}\) is necessarily nonsingular and its eigenvalues are exactly the non-zero eigenvalues of \(\mathbf{W}\) (that is, the non-zero eigenvalues of \(\mathbf{Q}^\top\)).

There are many possible choices for the matrix, \(\mathbf{T}\). 
Perhaps the most obvious choice for \(\mathbf{U}\) is the (\(N - 1\))-dimensional identity matrix, \(\mathbf{I}_{N-1}\), with a zero vector, \(\mathbf{0}\), appended to its right-hand side. In this case we have,
\begin{equation}
  \mathbf{T} =  \begin{pmatrix}
        1 & 0 & 0 & \ldots &0 & 0 \\
        0 & 1 & 0 &  \ldots & 0 & 0 \\
        0 & 0 & 1 & \ldots & 0 & 0 \\
        \vdots & \vdots & \vdots & \ddots & \vdots & \vdots \\
        0 & 0 & 0 & \ldots & 1 & 0 \\ 
        1 & 1 & 1 & \ldots & 1 & 1 
    \end{pmatrix}
    , \label{eqn:T_example1}
\end{equation}
from which obtain Equation~\eqref{eqn:app_simplified_ode1}. 
Note that for a reversible model, the transition-rate matrix \(\mathbf{Q}\) has only real eigenvalues, and we use a real-valued transformation matrix \(\mathbf{T}\) (such as the one defined in Equation~\eqref{eqn:T_example1}), we guarantee that \(\mathbf{C}\) has real eigenvalues, and so the system can be solved using diagonalisation with only real-valued matrices \citep{moler2003nineteen}. 
This is not the case when non-reversible models are used, in which case, it is possible that the resulting eigenvectors and eigenvalues are non-real, as discussed in Section~\ref{sec:app_reversibility}.

Now, instead of eliminating the \(N\)\textsuperscript{th} state, \(\mathbf{x}_N\), we could eliminate any other state. 
For example, to eliminate the first state, \(x_1\) instead of \(x_N\), we may define,
\begin{equation}
    \mathbf{T} = \begin{pmatrix}
        0 & 0 & 0 & \ldots & 0 &  1 \\
        0 & 1 & 0 &  \ldots & 0 &  0 \\
        0 & 0 & 1 & \ldots & 0 & 0 \\
        \vdots & \vdots & \vdots & \ddots & \vdots & \vdots \\
        0 & 0 & 0 & \ldots & 1 & 0 \\
        1 & 1 & 1 & \ldots & 1 & 1 
    \end{pmatrix}
    .
\end{equation} 

However, there are many other possible choices of \(\mathbf{T}\) which do not correspond to the removal of a single state.
It may be possible to choose \(\mathbf{T}\) such that the resultant system of ODEs has some desirable properties. 
%For example, we may apply some heuristic with the aim of reducing the error in matrix exponential computations. For example, we could seek a matrix \(\mathbf{T}\) that minimises the condition number of \(\mathbf{W}\) \citep{moler2003nineteen}.
We explore the effect of some possible choices \(\mathbf{T}\) in the following section.

\clearpage
\section{Examples} \label{sec:examples}
We introduce examples of the application of these results to two mathematical biology models. 
In both examples, the model in question contains a Markov model, but also includes some other elements. 
In the first example, the Markov model has transition rates which depend on an external variable, and the state occupancies are mapped to some observable output via an \emph{observation function}; in the second the Markov model is embedded within a larger deterministic chemical reaction network model. 
In both examples, we use the results introduced in Section~\ref{sec:app_steady_states}.

\subsection{Example I: Simulation of ion channel dynamics} \label{sec:example_ion_channel}
There are many Markov models which have been proposed to describe ion channel dynamics. 
In this section, we consider a selection of models of the \emph{rapid delayed rectifier} potassium current (\(I_\text{Kr}\)). 
This current plays an important role in the electrophysiology of various cells, including heart-muscle cells \citep{sanguinetti_herg_2006}. 
Deterministic, ODE-based Markov models (like those described in the previous sections) are commonly used to model the effect of thousands of channels carrying a current which helps to {repolarise} a cell after an action potential.

As in the previous sections, such a Markov model can be formulated as a system of ODEs in the form of Equation~\eqref{eqn:governing_equation}. Here, the state vector \(\mathbf{x}\) represents the proportion of channels in each of a collection of possible conformational states, and the transition-rate matrix \(\mathbf{Q}\) describes the rate at which the population of channels transitions between these states. 
The current, I\textsubscript{Kr}, is voltage-gated, meaning that these transition rates are dependent on an external variable, the transmembrane potential, \(V\), which itself can change through time,  $t$. 
To show this dependence explicitly we rewrite Equation~\eqref{eqn:governing_equation}, 
\begin{equation}
    \odv{x}{t} = \mathbf{Q}(V(t))^\top \mathbf{x} \, . \label{eqn:governing_equation_cardiac}
\end{equation}
We further require that for any \(i, j\) such that \(i\neq j\), we have either \(Q_{i, j}(V) = 0\) for all \(V\) or \(Q_{i, j} > 0\) for all \(V\). In other words, whilst the transition rates themselves may change, the fundamental graph properties of the model are independent of the transmembrane potential.

Such models may be incorporated into whole-cell action-potential models \citep{rudy_computational_2006}. 
Often, such models are fitted to data obtained from patch-clamp experiments performed on cells during which the current of interest is pharmacologically isolated \citep{wang_quantitative_1997,beattie_sinusoidal_2018, lei_rapid_2019-1}. In such experiments, this current may be recorded whilst the transmembrane potential, \(V\), is manipulated. 

As the experiments are performed after the cell has equilibrated, it is assumed that at the beginning of the experiment, where \(t=0\), we have,
\begin{equation}
    \mathbf{x}(0) = \mathbf{x}_\infty(V_0),
\end{equation}
where \(\mathbf{x}_\infty(V_0)\) is the steady state, guaranteed to exist (for any \(V\)) by Proposition~\ref{thm:unique_global_equilibrium}. 
This fact underpins typical inference approaches \citep{clerx_four_2019}, where transition-rate parameters are inferred from data, without considering the system's initial conditions as additional independent parameters. 
For example, when inferring transition rates from current recordings obtained from voltage-clamp experiments \citep{fink_markov_2009, beattie_sinusoidal_2018, clerx_four_2019}. 
Here, each proposed parameter set in an optimisation \(\mathbf{x}_\infty(V_0)\) is computed and used as the initial condition (for a forward simulation and subsequent objective function evaluation), which should be a good approximation for the conditions within experiments where the cells are held at $V_0$ for an extended period before recordings begin.
An observation function is then used to compare the Markov model and observed data, typically of the form,
\begin{equation}
    I_\text{Kr} = gx_\text{O}(V - E_\text{Kr}),
\end{equation}
where \(g\) is a further parameter known as the maximal conductance, \(x_\text{O}\) is the state variable corresponding to a single \emph{open state}, \(V\) is the transmembrane potential (controlled by the experimenter) and \(E_\text{Kr}\) is an additional constant (the reversal potential, often computed from first principles with the Nernst equation). 
For any Markov model satisfying Conditions~\eqref{eqn:Q_conditions_start}--\eqref{eqn:app_Q_conditions_end}, a stable equilibrium is assumed at time \(t=0\) which removes the need to infer the initial condition, \(\mathbf{x}(0)\). 
Instead, we only need to infer the transition-rate parameters and \(g\).
Whilst only one model structure is used here, the same approach may be taken with many other proposed Markov model structures, such as those discussed in \citet{mangold2021identification, shuttleworth_evaluating_2025}. 

\begin{figure}[htbp]
    \centering
   \includegraphics{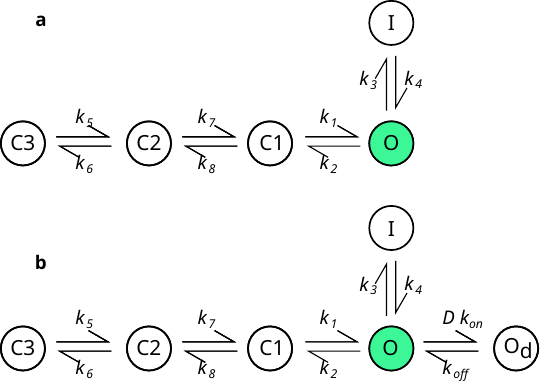}
    \caption{Markov models of I\textsubscript{Kr} used in this section. \textbf{a}: the unaltered \citet{wang_quantitative_1997} Model with no drug-bound state(s). \textbf{b} the \citet{wang_quantitative_1997} model with an additional drug-bound state. In the drug-bound model, there are two additional transition rates \(D k_\text{on}\) and \(k_\text{off}\). In this way, when \(D=1\), the model is strongly connected and hence irreducible. If \(D=0\), the model is not irreducible, but channels in the drug-bound state, \(O_D\), will transition into the other, drug-free states. However, the existence of a globally stable equilibrium is still guaranteed for the drug-bound model as well as the drug-free model. 
    }
    \label{fig:IKr_models}
\end{figure}

\subsubsection{Reducible Markov models in Cardiac Electrophysiology}
Other models used to model voltage-gated currents like I\textsubscript{Kr} include so-called Hodgkin-Huxley gating models. These draw their inspiration from Hodgkin and Huxley's seminal action-potential model \citep{hodgkin_quantitative_1952-1} which includes \emph{gating variables} of the form,
\begin{align}
    \odv{m}{t} &= \alpha_m (1 - m) - \beta_m m, \\ 
   \intertext{and} \odv{h}{t} &= \alpha_h (1 - h) - \beta_h h, 
\end{align} 
which combine to form Hodgkin and Huxley's sodium current, \(I_\text{Na}\) which is proportional to \(m^3 h\). 
Such a model can be seen as the combination of two separate, irreducible Markov models, though the existence of a unique stable equilibrium may also be proved directly \citep{keener2009mathematical}. 
Products of Hodgkin-Huxley gates may be combined into a single equivalent irreducible Markov model with a higher number of states \citep{rudy_computational_2006}, as done for I\textsubscript{Kr} in \citet{clerx_four_2019} (using a Hodgkin-Huxley style formulation) and \citet{beattie_sinusoidal_2018} (using a single irreducible Markov model). 
Either way, the results of the preceding sections relating to the existence and uniqueness of steady states, reversibility, and the related methods also apply to Hodgkin-Huxley style ion channel gating models.

\subsubsection{Drug binding models}
Markov models of cardiac ion-channel currents are commonly augmented with drug-block components \citep{li_improving_2017, lei2024impact}. 
Additional states and transition-rates are included, representing the way in which certain compounds can modify the kinetics of an ion channel. 
Such models are comprised of two components. 
A subnetwork modelling kinetics of ion channels in the absence of drugs (for example, those Markov models shown previously in this section), and a subnetwork that models the kinetics of drug block. 
Such an example is shown in Figure~\ref{fig:IKr_models}. 
The additional states representing the binding of drugs to an ion channel are included in the model as extra rows and columns in the transition-rate matrix, and, accordingly, additional components of \(\mathbf{x}\).
Note that some of these transition rates are dependent on the concentration of drug present, denoted by \(D\).

However, the inclusion of these additional states and transition rates may have an effect on the properties of the model. For instance, the model is reducible when \(D=0\), because there is no path to the drug-bound state from any other state in the model. However, a unique equilibrium can still be computed when \(D=0\) by considering only the Markov model representing the drug-free kinetics of ion-channel gating (since no channels can be drug-bound state at equilibrium). To see why this is the case, we may consider the two strong-connected components of the model when \(D=0\): the drug-bound state (denoted \(x_D\)), and drug-free states. We see that there is no flux between the latter to  the former when \(D=0\), but the flux in the opposite direction is positive provided \(x_D > 0\). Hence as \(t \rightarrow \infty\) we have \(x_D = 0\). It follows that as \(\mathbf{x} \rightarrow \mathbf{x}_\infty\) as \(t \rightarrow \infty\) where \(\mathbf{x}_\infty\) is identical to the steady state of the drug-free model, but with an additional zero-valued entry corresponding to the drug-bound state.

%For instance, while the \citet{kemp_electrophysiological_2021} Model shown in Figure~\ref{fig:IKr_models} is reversible, the drug-augmented model shown in Figure~\ref{fig:IKr_models} does not automatically satisfy the Condition~\eqref{eqn:app_kolmogorov_condition}, and so, is not a reversible model. 

When this is the case, it is clear that for all \(i\) corresponding to a drug-bound state, we have \(x_i \rightarrow 0\)  as \(t \rightarrow \infty\). Consequently, we can consider a Markov model with the drug-bound states removed, and providing this model is irreducible, we can see that a unique stable equilibrium exists through Proposition~\ref{thm:unique_global_equilibrium}.
On the other hand, when \(D > 0\), we can consider the full Markov model in its entirety.
Because Conditions~\eqref{eqn:Q_conditions_start}--\eqref{eqn:app_Q_conditions_end} are satisfied and the digraph is strongly connected (see Section~\ref{sec:connectedness}), we know that a unique steady state exists. That is to say, the results introduced in Section~\ref{sec:app_steady_states} may be applied directly to this special case, even though the complete model is not irreducible, strictly-speaking.

To demonstrate this fact, we choose the \citet{wang_quantitative_1997} model of I\textsubscript{Kr} and augment this with a single `drugged state' connected to channel model's open state. 
However, these methods can be applied to any number of alternative ion-channel model topologies and drug-binding schemes such as those as discussed in \citet{mangold2017mechanisms, lei2024impact, patten2025optimizing}. 

By considering the detailed balance equations at equilibrium or by applying Theorem~\ref{thm:kolmogorov_reversibility} (noting that there are no cycles in the model), we see that the \citet{wang_quantitative_1997} model is reversible whether the drug is present or not (that is, \(D>0\) or \(D=0\)). Thus, in either case, Proposition~\ref{prop:reversible_real_eigenvalues} guarantees that the transition-rate matrix, \(\mathbf{Q}\), has only real eigenvalues and, consequently, that this model cannot produce damped oscillatory solutions.

We may simulate this model using the model-reduction methods described in Section~\ref{sec:computational}.  We simulate this model using \texttt{scipy}'s ``LSODA'' routine which adaptively switches integration method when it detects stiffness/non-stiffness in the solution \citep{postawa2020comprehensive}. Absolute and relative solver tolerances were set to \(10^{-12}\). 
The parameters used were those presented by \citet{wang_quantitative_1997} corresponding to a 2\,mM concentration of KCl.
A simulation of this model under various drug concentrations (\(D\)) is shown in Figure~\ref{fig:drug_model_simulation}, where we see, after starting from some non-equilibrium initial conditions, the model decays to some equilibrium with and without a positive drug concentration, as we would expect; such equilibria are unique, and are guaranteed to exist as discussed above.

We set the drug-binding parameter \(k_\text{on}=10^{-1}\), and the corresponding unbinding parameter, \(k_\text{off}=10^{-2}\). 
Note that for real-life applications, these parameters would be fitted to represent the binding kinetics of a particular drug, whereas the other model parameters should be independent of the particular drug being considered.

\begin{figure}[htbp]
    \centering
    \includegraphics{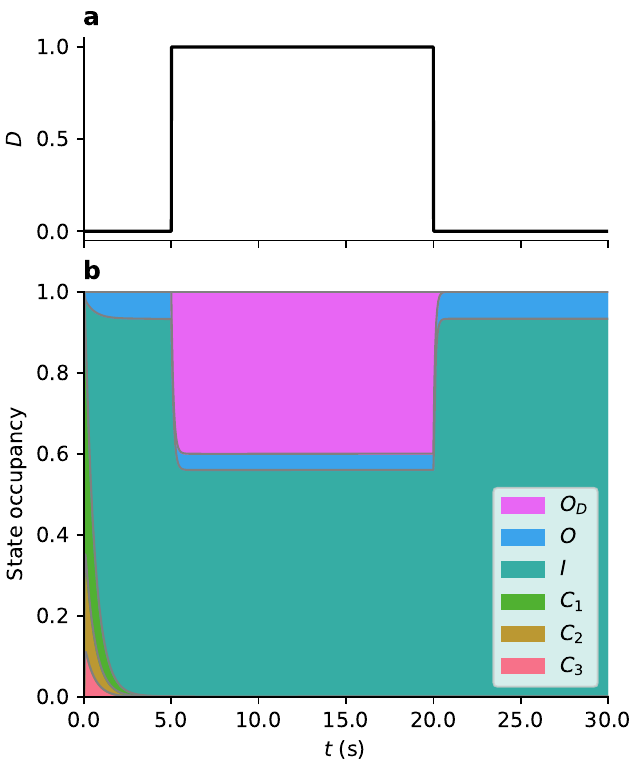}
    \caption{
    Stable equilibria for the drug-bound ion-channel model exist with and without the presence of drugs. Simulation of the \citet{wang_quantitative_1997} Model augmented with a drug-bound open state. In the absence of drug \(D=0\), the model can be seen as a reducible Markov model, where channels in the drug-bound state slowly unbind and transition to drugless states. For \(D=1\), the model topology is strongly connected and so, may simply be treated as a single irreducible Markov model. 
    \textbf{a}: time-course showing the application of drug when \(1 \leqslant t \leqslant 3\) (\(D=1\)), and absence of drug \(D=0\) otherwise. 
    \textbf{b}: The proportion of channels in each state at time \(t\) according to the model. Two steady states are clearly shown: one when \(D=0\) (without drug), and one where \(D=1\) (with drug). Moreover, once the drug is removed, the system returns to the drugless steady state. Whether \(D=0\) or \(D>0\) these steady states are unique, as discussed above.
    }
    \label{fig:drug_model_simulation}
\end{figure}

\subsubsection{Removing a state variable}
In Section~\ref{sec:computational}, we described how the governing Equation~\eqref{eqn:governing_equation} may be simplified using different choices of a particular transformation matrix (denoted \(T\)). By mapping a state, or linear combination of states, to \(0\), this allows the governing equation to be rewritten with one fewer state variables. Here, we explore the numerical properties of the Wang model when transformed by various \(\mathbf T\) matrices. Firstly, we rewrite the model in \(5\) different ways by removing each of the model's five state variables. We then compare these approaches against the full form of the model using the unaltered \(\mathbf{Q}\)-matrix (as in Equation~\eqref{eqn:governing_equation}).

Again, use the \citet{wang_quantitative_1997} model, and simulate it using \texttt{scipy}'s ``LSODA'' routine. For each simulation, relative and absolute solver tolerances were set to the same value.
Results of this approach are shown in Figure~\ref{fig:model_simplification_approach}, and we perform simulations for a range of values. Here, it seems all the different reductions schemes are similarly suitable---in each case, the simplified model requires far fewer steps than the full model to achieve a similar level of accuracy. Though it remains to be seen if different choices of \(\mathbf{T}\) require less solver steps for this, and other, example models. 

\begin{figure}[htbp]
    \centering
    \includegraphics{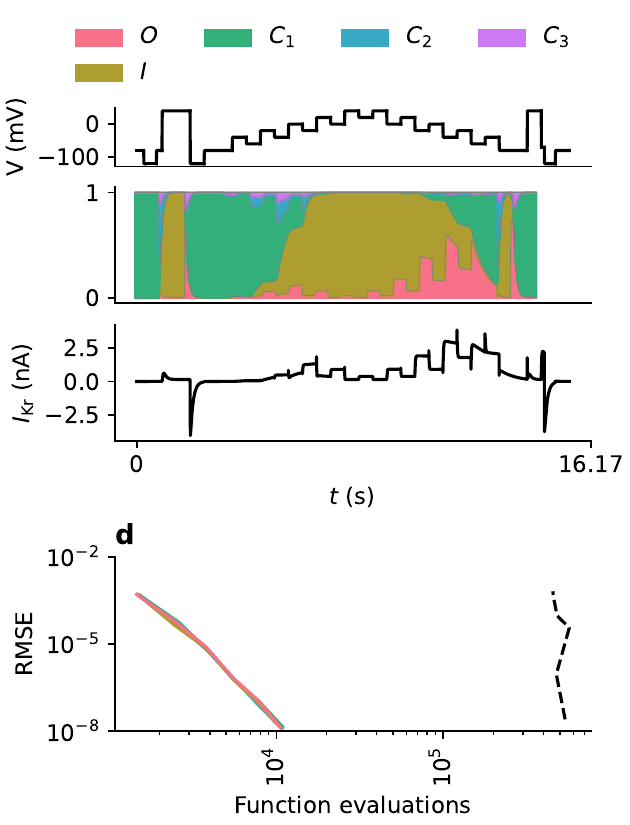}
    \caption{
    Regardless of which state is eliminated from the model, significantly less function evaluations are required to simulate the \citet{wang_quantitative_1997} model. \textbf{a} A staircase-like voltage signal used in \citet{lei_rapid_2019-1} is input to the model. 
    \textbf{b}: the proportion of ion channels in each of the five states at time \(t\). 
    \textbf{c}: The resulting current \(I_\textrm{Kr}\) with the maximal conductance, \(g_\text{Kr}\), set to \SI{0.1524}{\micro\siemens}. 
    \textbf{d}: the number of function evaluations required to simulate the model and the corresponding root-mean-square error (RMSE) between the solution and a reference solution (for which both the relative and absolute solver tolerances were set to \(10^{-12}\)). For a range of tolerances between \(10^{-2}\) and \(10^{-8}\), we show the computed error between each version of the model (with a given state removed, or computed using the entire \(\mathbf{Q}\) matrix) .
    Here, eliminating each state (those shown in the legend) at the top of the figure  results in a similar amount of error
   (compared to a reference solution obtained with fine tolerances).
   \textbf{e}: The error when eliminating a state (here, state O) is significantly less than when the model is simulated using the full transition-rate matrix, with no state variables removed (black dashed line).
   }
    \label{fig:model_simplification_approach}
\end{figure}

\subsection{Example II: Exploring the fast-slow dynamics of an auxin signalling pathway} \label{sec:example_auxin}
Nuclear signalling pathways play an important role in the response of living cells (both plant and animal) to environmental cues. 
The modelling of such signalling pathways is a ubiquitous challenge encountered in many areas of mathematical biology \citep{klipp_mathematical_2006}. 
Models can quickly become infeasibly complicated when large signalling pathways are considered in their entirety. 
Here, we present an example of a signalling pathway model, where the presence of a Markov-model subnetwork allows for model simplification.

Auxin is a plant hormone which, amongst other things, controls the expression of genes relating to cell growth and maturation \citep{leyser_auxin_2018}. 
Auxin interacts with the signalling pathway by promoting the degradation of Aux/IAA inhibitor proteins. Our model describes a concentration of Aux/IAA proteins, denoted \(I\), and a concentration of Auxin Response Factors (ARFs), denoted \(A\).
As in previous work, such variables may represent either specific members of the Aux/IAA and ARF gene families, or some combinations of genes \citep{middleton_mathematical_2010, farcot_modular_2015}.

Individual ARF and Aux/IAA proteins are known to bind together to form \emph{dimers}. We denote the concentration of a dimer formed by a single ARF and Aux/IAA as \(D_{A, I}\).
For the sake of parsimony, we assume that \(D_{A, I} = D_{I, A}\). We also consider ARF-ARF homodimers with concentration denoted by \(D_{A,A}\), and Aux/IAA-Aux/IAA homodimers with concentrations denoted by \(D_{I, I}\).
Then, \(D_{A, A}\),  \(D_{A, I}\), and \(D_{I, I}\) are all state-variables representing the concentrations of dimers in the nucleus. 

Following \cite{shuttleworth_general_2026} and assuming that all reactions between the various species (that is proteins and protein complexes) follow the mass-action law, we have,
\begin{align}
    \odv{A}{t}&= k_A - d_A A - 2A^2k_{A, A} - k_{A, I} A I + 2 d_{A, A}D_{A, A}   + d_{A, I}D_{A, I},\label{eqn:arfA_ode}\\
    \odv{I}{t} &= k_I - d_I I - AIk_{A, I}  - 2I^2 k_{I, I} + 2d_{I, I}D_{I, I} + d_{A, I} D_{A, I}, \label{eqn:I_ode} \\
    \odv{D_{A, A}}{t} &= A^2 k_{A, A} - d_{A, A} D_{A, A}, \label{eqn:d_aa_ode}\\
    \odv{D_{A, I}}{t} &= A I k_{A, I} - d_{A, I} D_{A, I}, \label{eqn:d_ai_ode} \\
    \odv{D_{I, I}}{t} &=  - d_{I , I} D_{I, I} + I^2 k_{I, I} \label{eqn:d_ii_ode}
\end{align}
where \(k_{A, A}, k_{A, I},\)  and \(k_{I, I}\) are association rates that control the rate at which dimers are formed; \(d_{A, A}, d_{A, I},\)  and \(d_{I, I}\) are the corresponding dissociation rates; and \(k_A\) and \(d_A\), and \(k_I\) and \(d_I\) are the production and decay rates for ARFs and Aux/IAA proteins, respectively. 
The production and decay rates for ARFs, that is \(k_A\) and \(d_A\), and the Aux/IAA production rate are assumed to be constant, whereas Aux/IAA decays more rapidly in the presence of auxin.
This fact is what enables the signalling network to respond to changes in auxin concentration, because these changes in concentration affect the rate at which target genes are transcribed.

This response to auxin enables many important processes such as cell growth, division and differentiation, for example.
In this model, the effect of an increased/decreased auxin concentration in the cell can be achieved by increasing/decreasing the Aux/IAA decay parameter, \(d_I\), which is assumed constant.
In this section, we consider only the dynamics of the signalling pathway under steady-state concentrations of auxin---in which case, we assume that \(d_I\) is constant over the entire time-course.

It is natural to use a Markov model to describe the binding/unbinding of complex proteins to DNA. Here individual protein complexes (such as  individual proteins and dimers) may bind to and unbind from the DNA's promoter region at rates proportional to their respective concentrations. 
We write this part of the model as,
\begin{equation}
    \odv{\mathbf{g}}{t} = \dfrac{1}{\tau_G} \mathbf Q_\mathbf{g}(\mathbf{x}, \boldsymbol{\theta})^\top \mathbf{g}\,, \label{eqn:promoter_Q_equation}
\end{equation}
where \begin{equation}
    \mathbf{g} = \begin{pmatrix}
        G \\
        G_{A, A} \\ 
        G_{A, I}
    \end{pmatrix},
\end{equation} 
is a vector comprising the proportions of all possible promoter configurations (unbound, bound to an ARF-ARF dimer and bound to an ARF-Aux/IAA, respectively), \(\mathbf{Q}_\mathbf{g}\) is a transition-rate matrix which depends on the models state (\(\mathbf{x}\)) and parameter (\(\boldsymbol{\theta}\)) vectors, and \(\tau_G\) is some representative timescale that dictates the overall rate of the dynamics, but does not affect the steady state.
For fixed \(\mathbf{x}\) and \(\boldsymbol{\theta}\), this portion of the model is a Markov model as described in Section~\ref{sec:introduction}.
This transition-rate matrix may be parameterised so that it reflects affinities for protein-promoter kinetics. 

The above model equations, along with the computer code implementation of the model, were produced using the \textit{auxin\_signalling\_models} package (v0.3.1) which may found on GitHub. Further details are provided in Section~\ref{sec:data_availability}

This Markov model then describes the probability that a given promoter is bound to one of each of the model's monomers or dimers at a given time. 
These promoter-state variables dictate the transcription rates which ultimately control the production of ARF and Aux/IAA monomers, as described in Equations~\eqref{eqn:arfA_ode} and \eqref{eqn:I_ode}. Conversely, the concentrations of ARFs, Aux/IAAs and dimers thereof dictates the transition rates between states in the promoter Markov model. This interaction and feedback between the Markov-model component of the model, and the concentration of ARF and Aux/IAA proteins, is shown in Figure~\ref{fig:auxin_diagram}. 

\begin{figure}[htbp]
    \centering
    \includegraphics{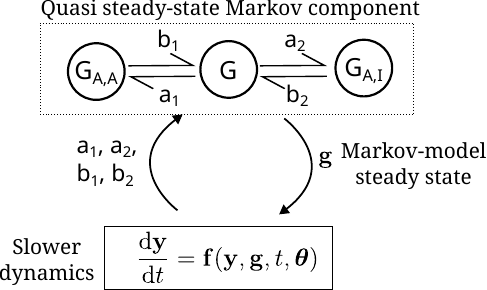}
    \caption{Schematic showing a Markov model as a component of a larger ODE model. When the transition rates in the Markov model are fast compared to other dynamics of the model, a quasi steady-state assumption may be used. In this case, the unique steady state of the Markov model is determined by the \(\mathbf{y}\) state variables described by the rest of the model. In turn, the quasi steady-state variables the model is used to compute \(\odv{\mathbf{y}}{t}\) to simulate the model. For the auxin model described in the article, the Markov model component corresponds to the binding of complexes with the promoter region of the DNA, and the remaining state variables govern the production, decay and dimerisation of proteins, and the production and decay of mRNA.
    }
    \label{fig:auxin_diagram}
\end{figure}

Each promoter state corresponds to a different transcription rate for each species of mRNA. There are two types of mRNA in our model, with concentrations denoted by \(R_A\) and \(R_I\), that drive the expression of ARF and Aux/IAA proteins, respectively. Each  mRNA concentration, \(R_i\), is governed by an ODE of the form, 
\begin{align}
    \odv{R_i}t &= K_i(\mathbf{g}, \boldsymbol{\theta}) - d_{R} R_i,
\end{align}
where \(d_{R}\) is the assumed (time-independent) decay rate for mRNA (assumed equal for both types of mRNA described by the model), and \(K_i(\mathbf{g}, \boldsymbol{\theta})\) is a function which maps the promoter states and parameter vectors to some transcription rate,
\begin{equation}
K_i(\mathbf{g}, \boldsymbol{\theta}) = K_{\text{max}, i} G_{A, A} + K_{0, i} G,
\end{equation}
where \(K_{\text{max}, i}\) is a constant parameter which quantifies the rate of protein translation when an ARF-ARF homodimer is bound to the promoter, and \(K_{0, i}\) is the basal rate of transcription which occurs when no ARF nor Aux/IAA protein is bound to the promoter. 
Here, we assume that ARF-Aux/IAA dimers have a repressive effect on both ARF and Aux/IAA transcription and consequently, the promoter state \(G_{A, I}\) has no associated transcription rate.
Such a motif is found in Arabidopsis thaliana (particularly in early embryos), consisting of  ARF, ARF5/Monopteros (MP) and the Aux/IAA, IAA12/Bodenlos (BDL) \citep{lau_auxin_2011}.

 For small \(\tau_G\), where promoter binding/unbinding occurs relatively quickly, we can approximate the system using a \emph{quasi steady-state assumption} to define the \emph{slow-manifold model},
\begin{align}
    \label{eqn:reduced_model_start}
    \odv{\mathbf{y}}{t} &=\mathbf{f}(\mathbf{y}, \mathbf{g}, t; \boldsymbol{\theta}),\\
    \mathbf{g} &=\mathbf G_\infty(\mathbf{y}, \boldsymbol{\theta}), \label{eqn:reduced_model_end}
\end{align}
where \(\mathbf{f}(y, t; p)\) is the vector-valued function describing the evolution of 
the state variables as in Equation~\eqref{eqn:arfA_ode}--\eqref{eqn:d_ii_ode}, and the function \(\mathbf{G}\) maps these state variables, \(\mathbf{y}\), and our parameter vector, \(\boldsymbol{\theta}\) to the corresponding equilibrium point of the promoter Markov model. 
This equilibrium point is determined by finding the non-trivial solution to Equation~\eqref{eqn:promoter_Q_equation} with elements summing to one, that is, \(\mathbf{1}^\top \mathbf{g} = 1\).  Such an equilibrium point is guaranteed to exist by Proposition~\ref{thm:unique_global_equilibrium}.
The simpler slow-manifold model is then a suitable approximation for the full model when \(\tau_G\) is small \citep{battelli_pseudo-steady-state_1985}.
The assumption that these protein-promoter interactions occur on a much faster timescale than the rest of the model is used in previous models which do not include an explicit description of promoter-protein binding dynamics \citep{middleton_mathematical_2010, bridge2012distinguishing, farcot_modular_2015}.
Nevertheless, simplifying the model in this way reduces the dimensionality of the dynamical system.

Furthermore, we may find the steady state of the promoter Markov model in the above model explicitly,
\begin{align}
    G &= \dfrac{1}{1 + \lambda_{A, A} + \lambda_{A, I} },\, \\ 
    G_{A, A} &= \dfrac{\lambda_{A, A}}{
    1 + \lambda_{A, A} + \lambda_{A, I}
    },\,  \\ 
   G_{A, I} &= \dfrac{\lambda_{A, I}}{
    1 + \lambda_{A, A} + \lambda_{A, I}}\,,
\end{align}
where \(\lambda_{i, j}\) are ratios of transition rates in the Markov model component,
\begin{align}
    \lambda_{A, A} &= D_{A, A} \dfrac{k_{G, D_{A, A}}}{d_{G, D_{A, A}}}\,, \\
    \lambda_{A, I} &= D_{A, I} \dfrac{k_{G, D_{A, I}}}{d_{G, D_{A, I}}}\,.
\end{align}
Hence, in this case, the steady state of the Markov model depends only on the ratio of pairs of transition rates. This fact allows for a reduction in the number of necessary model parameters under the quasi steady-state assumption outlined above.

Figure~\ref{fig:degenerate_markov_model} shows the behaviour of the model with and without the quasi steady-state assumption, and under the parameter values given in Table~\ref{tab:asp_model_parameters_full}. This figure also shows how the resultant model trajectories converge to the quasi steady-state approximation model as \(\tau_G\) becomes small. Whilst the error between the full model and the approximate model decreases for smaller \(\tau_G\), the number of RHS evaluations required by the solver increases due to the increasing stiffness of the system. Accordingly, we see that for small \(\tau_g\), the quasi steady-state model provides a good approximation of the full model, and its usage reduces computational expense.

%\subsubsection{Full model parameters}
%Dict("d_{R_{}}" => 0.1, "k_{G_{ARF_1, R, iaa_1}}" => 1.0, "k_{D_{iaa_1, iaa_1}}" => 0.05, "k_{D_{ARF_1, ARF_1}}" => 0.05, "d_{D_{ARF_1, iaa_1}}" => 0.01, "k_{G_{ARF_1, ARF_1}}" => 0.1, "d_{D_{ARF_1, ARF_1}}" => 0.01, "k_{G_{ARF_1, iaa_1, R, iaa_1}}" => 0.0, "k_{G_{ARF_1, R, ARF_1}}" => 1.0, "d_{arf_{all}}" => 2.0, "k_{G_{ARF_1, ARF_1, R, ARF_1}}" => 1.0, "d_{iaa_{all}}" => 0.1, "d_{G_{ARF_1, ARF_1}}" => 0.01, "k_{G_{R, iaa_1}}" => 0.01, "k_{D_{ARF_1, iaa_1}}" => 0.05, "k_{G_{ARF_1, iaa_1, R, ARF_1}}" => 0.0, "d_{iaa_1_{x}}" => 1.0, "k_{G_{ARF_1, ARF_1, R, iaa_1}}" => 1.0, "d_{D_{iaa_1, iaa_1}}" => 0.01, "k_{RNA_{ARF}}" => 1.0, "k_{RNA_{IAA}}" => 1.0, "k_{G_{R, ARF_1}}" => 0.0001, "k_{G_{ARF_1, iaa_1}}" => 0.1, "d_{G_{ARF_1, iaa_1}}" => 0.01)

% The model parameters used for the simulation and analysis of this model in are provided in Table~\ref{tab:asp_model_parameters_full}.

\begin{table}[htbp]
    \centering
    \begin{tabular}{l c c}
    \toprule
    Related process & Parameter & value \\
    \midrule
   Dimerisation &    \(k_{I, I}\) & 0.15\\
  &  \(d_{I, I}\) & 0.01 \\
  &  \(k_{A, A}\) & 0.05\\
  &  \(d_{A, A}\) & 0.01 \\
  &  \(k_{{A, I}}\) & 0.05\\
  &  \(d_{A, I}\) & 0.01 \\
  \midrule
Monomer degradation  &   \(d_A\) & 2.0 \\
  &  \(d_{I}\) & 0.1 \\
 Monomer production &  \(k_{{A}}\) & 1.0 \\
  &  \(k_{{I}}\) & 1.0\\
  \midrule
  RNA degradation
  &  \(d_{R}\) &  0.1 \\
  \midrule
 Aux/IAA mRNA transcription 
   &  \(K_{\text{max}, R_I}\) & 1.0 \\
  &  \(K_{{0, R_I}}\) & 0.01 \\
  \midrule
   ARF mRNA transcription 
  &  \(K_{\text{max}, R_A}\) & 1.0 \\
  &  \(K_{0, R_A}\) & 0.0001\\
    \bottomrule
    \end{tabular}
    \caption{List of model parameters and associated values used to simulate and analyse our model of the auxin-signalling pathway. All  promoter-protein  association/dissociation rates  are set to \(1\). Note that a different transcription rate is given for each state of the promoter in the model. The effective transcription-rate is obtained by taking an average of these rates, weighted by the probability that the promoter is in each state.}
    \label{tab:asp_model_parameters_full}
\end{table}

\begin{figure}[htbp]
    \centering
    \includegraphics[width=\textwidth]{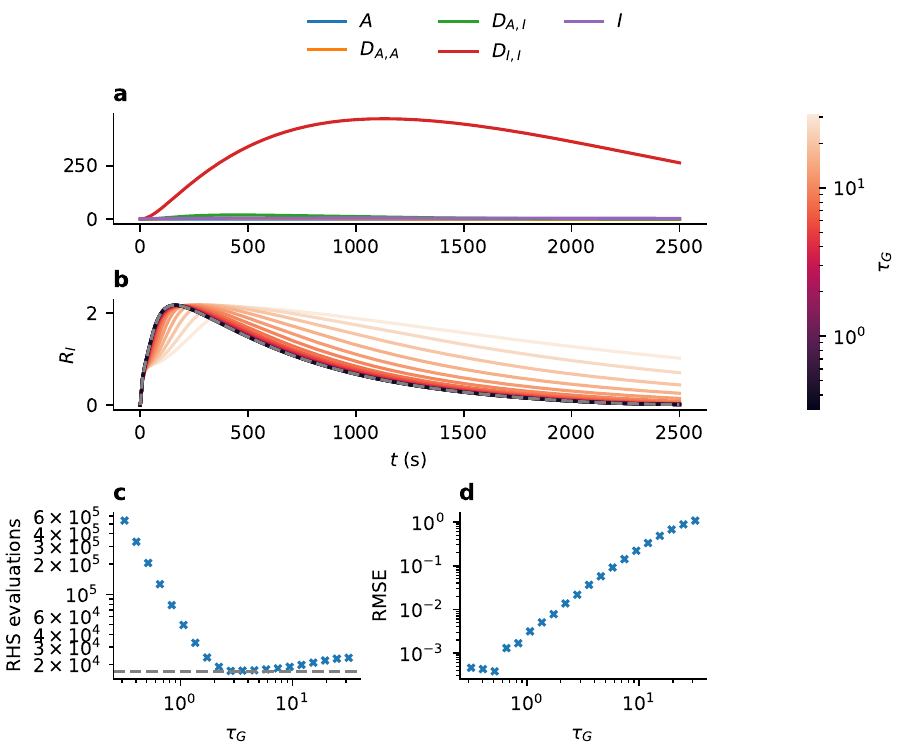}
    \caption{
    The full model becomes well-approximated by the slow-manifold model as the promoter timescale, \(\tau_G\) decreases. 
    Simultaneously the full model becomes more expensive to simulate as \(\tau_G\) decreases. 
    Panel \textbf{a}: a simulation of the slow-manifold model. 
    Panel \textbf{b}: trajectories of the full model for various values of \(\tau_G\) showing the Aux/IAA mRNA concentration, \(R_\text{I}\). 
    Here the output of slow-manifold model is shown by the grey, dashed line. 
    Panel \textbf{c}: a scatterplot showing how the number of RHS evaluations required to compute the above trajectories depends on \(\tau_G\). 
    The number of RHS evaluations required by the slow-manifold model is shown by the horizontal, grey, dashed line. 
  \textbf{d}: root-mean-square error between the above \(R_I\) time series from the slow-manifold model and the full model for various values of \(\tau_G\).}
    \label{fig:degenerate_markov_model}
\end{figure}

This reduction in the number of state variables also proves useful when exploring the steady-state behaviour of such models \citep{shuttleworth_general_2026}.

\section{Discussion}
Mathematical models may be used in a variety of settings to provide mechanistic and quantitative descriptions of complex biological phenomena.
For certain applications, these models may be expressed as irreducible systems of linear ODEs. Such mathematical models are guaranteed certain properties---chiefly the existence and uniqueness of stable equilibria.
We have mainly discussed models where the derivative function may be expressed in terms of an irreducible matrix. However, as in the case of the drug-binding model discussed in Section~\ref{sec:example_ion_channel}, similar results may also apply to closely-related linear systems of ODEs.
We also saw in Section~\ref{sec:app_reversibility} that microscopic reversibility has consequences for model implementations in code, where eigenvalues/eigenvectors should be assumed to be real-valued only if
the model satisfies microscopic reversibility.

We have also seen that such formulations allow for a reduction in model complexity (number of state variables) which facilitates mathematical investigation, and may lead to significant performance gains in computer code. Such performance gains are particularly valuable when fitting models---as are the aforementioned mathematical guarantees regarding steady states. The fact that the conservation condition allows for the removal of one state variable is well known \citep{colquhoun1995q, fink_markov_2009}, though the choice of which state (or linear combination of states) to remove is not obvious. In fact, we saw that this choice was of little importance for the example presented in Section~\ref{sec:example_ion_channel}: no matter which state was removed, a similar number of function evaluations were required to provide a given accuracy.
However, this conclusion may change in new settings, particularly in conditions where transitions between certain states are much faster than others.
Further work identifying schemes for model reduction which improve the numerical properties of the resulting simplified ODE systems could further improve the implementation of Markov models. 

Our second example demonstrates how the mathematical properties of Markov models allows for model simplification of signalling pathway models, which typically include protein-protein binding, and protein-DNA binding, two processes which operate at different time scales and in completely different concentration ranges. This justifies model simplification via a quasi steady-state assumption, allowing for more efficient simulation, a reduction in the number of necessary model parameters, and possibly easier analysis of the resulting dynamical systems. 
Here, a simple model was introduced for demonstrative purposes, but it would be natural to extend this model to include multiple ARF species and multiple Aux/IAA species \citep{bridge2012distinguishing}. In such a case, our quasi steady-state assumption is yet more useful, removing even more state variables from the ODE system. 

Other examples of irreducible Markov models embedded in larger ODE mathematical biology models include other models of gene regulations \citep{kar_control_2016}, and I\textsubscript{K1} in cardiac action-potential models \citep{luo_dynamic_1994, ten_tusscher_model_2004}, which can be interpreted as a similar quasi steady-state formulation.
For any such model, the examples presented above show how the mathematical properties of irreducible Markov models can be used to simplify models and analyse the steady-state behaviour of systems, even when the irreducible Markov model forms but a small subnetwork of some larger system. 

\section*{Data Accessibility}
\noindent
All codes used for model generation, simulation, and the production of figures is available on \href{https://github.com/joeyshuttleworth/irreducible_markov_models}{GitHub} at \url{https://github.com/joeyshuttleworth/irreducible_markov_models}.
~\label{sec:data_availability}

\section*{Acknowledgments}
This work was supported by the Wellcome Trust (grant no. 212203/Z/18/Z); and the Leverhulme Trust (grant no. RPG-2024-061).
GRM acknowledges support from the Wellcome Trust via a Senior Research Fellowship. 
JGS and EF acknowledge support from the Leverhulme Trust.
This research was funded in whole, or in part, by the Wellcome Trust [212203/Z/18/Z]. 
For the purpose of open access, the author has applied a CC-BY public copyright licence to any Author Accepted Manuscript version arising from this submission.

\clearpage
\bibliography{references}

@article{crooks2011thermodynamic,
	doi = {10.1088/1742-5468/2011/07/P07008},
  title={On thermodynamic and microscopic reversibility},
  author={Crooks, Gavin E},
  journal={Journal of Statistical Mechanics: Theory and Experiment},
  volume={2011},
  number={07},
  pages={P07008},
  year={2011},
  publisher={IOP Publishing}
}

@article{ederer2007thermodynamically,
	doi = {10.1529/biophysj.106.094094},
  title={Thermodynamically feasible kinetic models of reaction networks},
  author={Ederer, Michael and Gilles, Ernst Dieter},
  journal={Biophysical journal},
  volume={92},
  number={6},
  pages={1846--1857},
  year={2007},
  publisher={Elsevier}
}

@book{wilkinson2018stochastic,
author = {Wilkinson, Darren James},
address = {Boca Raton},
booktitle = {Stochastic modelling for systems biology},
edition = {Third edition.},
isbn = {9781138549289},
language = {eng},
lccn = {2018040390},
publisher = {CRC Press},
title = {Stochastic modelling for systems biology},
year = {2019},
}

@article{patten2025optimizing,
	doi = {10.1098/rsta.2024.0227},
  title={Optimizing experimental designs for model selection of ion channel drug-binding mechanisms},
  author={Patten-Elliott, Frankie and Lei, Chon Lok and Preston, Simon P and Wilkinson, Richard D and Mirams, Gary R},
  journal={Philosophical Transactions A},
  volume={383},
  number={2292},
  pages={20240227},
  year={2025},
  publisher={The Royal Society}
}

@article{gupta_stochastic_2011,
	title = {Stochastic {State} {Transitions} {Give} {Rise} to {Phenotypic} {Equilibrium} in {Populations} of {Cancer} {Cells}},
	volume = {146},
	issn = {0092-8674, 1097-4172},
	doi = {10.1016/j.cell.2011.07.026},
	language = {English},
	number = {4},
	urldate = {2026-05-12},
	journal = {Cell},
	author = {Gupta, Piyush B. and Fillmore, Christine M. and Jiang, Guozhi and Shapira, Sagi D. and Tao, Kai and Kuperwasser, Charlotte and Lander, Eric S.},
	month = aug,
	year = {2011},
	pmid = {21854987},
	pages = {633--644},
}

@article{kepler_stochasticity_2001,
	title = {Stochasticity in {Transcriptional} {Regulation}: {Origins}, {Consequences}, and {Mathematical} {Representations}},
	volume = {81},
	copyright = {https://www.elsevier.com/tdm/userlicense/1.0/},
	issn = {00063495},
	shorttitle = {Stochasticity in {Transcriptional} {Regulation}},
	url = {https://linkinghub.elsevier.com/retrieve/pii/S0006349501759498},
	doi = {10.1016/S0006-3495(01)75949-8},
	language = {en},
	number = {6},
	urldate = {2026-05-12},
	journal = {Biophysical Journal},
	author = {Kepler, Thomas B. and Elston, Timothy C.},
	month = dec,
	year = {2001},
	pages = {3116--3136},
}

@article{mangold2017mechanisms,
	doi = {10.1080/19336950.2017.1369637},
  title={Mechanisms and models of cardiac sodium channel inactivation},
  author={Mangold, Kathryn E and Brumback, Brittany D and Angsutararux, Paweorn and Voelker, Taylor L and Zhu, Wandi and Kang, Po Wei and Moreno, Jonathan D and Silva, Jonathan R},
  journal={Channels},
  volume={11},
  number={6},
  pages={517--533},
  year={2017},
  publisher={Taylor \& Francis}
}

@article{lei2024impact,
  title={The impact of uncertainty in hERG binding mechanism on in silico predictions of drug-induced proarrhythmic risk},
  author={Lei, Chon Lok and Whittaker, Dominic G and Mirams, Gary R},
  journal={British Journal of Pharmacology},
  doi={10.1111/bph.16250},
  volume={181},
  number={7},
  pages={987--1004},
  year={2024},
  publisher={Wiley Online Library}
}

@article{beattie_sinusoidal_2018,
	title = {Sinusoidal voltage protocols for rapid characterisation of ion channel kinetics},
	volume = {596},
	copyright = {© 2018 The Authors. The Journal of Physiology published by John Wiley \& Sons Ltd on behalf of The Physiological Society},
	issn = {1469-7793},
	doi = {https://doi.org/10.1113/JP275733},
	language = {en},
	number = {10},
	journal = {The Journal of Physiology},
	author = {Beattie, Kylie A. and Hill, Adam P. and Bardenet, Rémi and Cui, Yi and Vandenberg, Jamie I. and Gavaghan, David J. and Boer, Teun P. de and Mirams, Gary R.},
	year = {2018},
	pages = {1813--1828}
}

@article{fink_markov_2009,
	title = {Markov models for ion channels: versatility versus identifiability and speed},
	volume = {367},
	shorttitle = {Markov models for ion channels},
	url = {https://royalsocietypublishing.org/doi/10.1098/rsta.2008.0301},
	doi = {10.1098/rsta.2008.0301},
	number = {1896},
	journal = {Philosophical Transactions of the Royal Society A: Mathematical, Physical and Engineering Sciences},
	author = {Fink, Martin and Noble, Denis},
	month = jun,
	year = {2009},
	pages = {2161--2179},
}

@article{clerx_four_2019,
	title = {Four Ways to Fit an Ion Channel Model},
	volume = {117},
	issn = {0006-3495},
	doi = {10.1016/j.bpj.2019.08.001},
	language = {en},
	number = {12},
	journal = {Biophysical Journal},
	author = {Clerx, Michael and Beattie, Kylie A. and Gavaghan, David J. and Mirams, Gary R.},
	month = dec,
	year = {2019},
	pages = {2420--2437},
}

@article{ten_tusscher_model_2004,
	title = {A model for human ventricular tissue},
	volume = {286},
	issn = {0363-6135, 1522-1539},
	url = {https://www.physiology.org/doi/10.1152/ajpheart.00794.2003},
	doi = {10.1152/ajpheart.00794.2003},
	language = {en},
	number = {4},
	
	journal = {American Journal of Physiology-Heart and Circulatory Physiology},
	author = {ten Tusscher, K. H. W. J. and Noble, D. and Noble, P. J. and Panfilov, A. V.},
	month = apr,
	year = {2004},
	pages = {H1573--H1589},
}

@article{postawa2020comprehensive,
	doi = {10.1016/j.renene.2020.04.089},
  title={{A comprehensive comparison of ODE solvers for biochemical problems}},
  author={Postawa, Karol and Szczygie{\l}, Jerzy and Ku{\l}a{\.z}y{\'n}ski, Marek},
  journal={Renewable Energy},
  volume={156},
  pages={624--633},
  year={2020},
  publisher={Elsevier}
}

@book{BermanAbraham1979Nmit,
series = {Computer science and applied mathematics},
publisher = {Academic Press},
title = {Nonnegative matrices in the mathematical sciences},
isbn = {0120922509},
year = {1979},
language = {eng},
address = {New York ; London},
author = {Berman, Abraham and Plemmons, Robert J},
}

@article{yang2006imposing,
	doi = {10.1529/biophysj.105.071852},
  title={On imposing detailed balance in complex reaction mechanisms},
  author={Yang, Jin and Bruno, William J and Hlavacek, William S and Pearson, John E},
  journal={Biophysical Journal},
  volume={91},
  number={3},
  pages={1136--1141},
  year={2006},
  publisher={Elsevier}
}

@book{bates_nonlinear_1988,
	address = {New York},
	title = {Nonlinear regression analysis and its applications},
	isbn = {0-471-81643-4},
	language = {eng},
	publisher = {Wiley},
	author = {Bates, Douglas M and Watts, Donald G},
	year = {1988},
	lccn = {8806065},
}

@article{colquhoun_how_2004,
	title = {How to Impose Microscopic Reversibility in Complex Reaction Mechanisms},
	volume = {86},
	issn = {0006-3495},
	doi = {10.1529/biophysj.103.038679},
	number = {6},
	
	journal = {Biophysical Journal},
	author = {Colquhoun, David and Dowsland, Kathryn A. and Beato, Marco and Plested, Andrew J. R.},
	month = jun,
	year = {2004},
	pmid = {15189850},
	pmcid = {PMC1304255},
	pages = {3510--3518},
}

@article{mangold2021identification,
	doi = {10.1371/journal.pcbi.1008932},
  title={Identification of structures for ion channel kinetic models},
  author={Mangold, Kathryn E and Wang, Wei and Johnson, Eric K and Bhagavan, Druv and Moreno, Jonathan D and Nerbonne, Jeanne M and Silva, Jonathan R},
  journal={PLoS computational biology},
  volume={17},
  number={8},
  pages={e1008932},
  year={2021},
  publisher={Public Library of Science San Francisco, CA USA}
}

@book{keener2009mathematical,
series = {Interdisciplinary applied mathematics},
publisher = {Springer},
booktitle = {Mathematical Physiology},
isbn = {0387983813},
year = {2009},
title = {Mathematical Physiology},
language = {eng},
address = {New York},
author = {Keener, James P and Sneyd, James},
lccn = {lc98014499},
}

@book{gibbons_algorithmic_1985,
	address = {Cambridge, United Kingdom},
	title = {Algorithmic graph theory},
	isbn = {0-521-24659-8},
	language = {eng},
	publisher = {Cambridge University Press},
	author = {Gibbons, Alan},
	year = {1985},
}

@misc{shuttleworth_general_2026,
	title = {A general mathematical framework for modelling subnetworks of the nuclear auxin pathway},
	copyright = {© 2026, Posted by openRxiv. This pre-print is available under a Creative Commons License (Attribution 4.0 International), CC BY 4.0, as described at http://creativecommons.org/licenses/by/4.0/},
	url = {https://www.biorxiv.org/content/10.64898/2026.08.06.742982v1},
	doi = {10.64898/2026.08.06.742982},
	language = {en},
	urldate = {2026-08-10},
	publisher = {bioRxiv},
	author = {Shuttleworth, Joseph G. and Chan, Emily and Welch, Thomas and Bhosale, Rahul G. and Bishopp, Anthony and Farcot, Etienne},
	month = aug,
	year = {2026},
	note = {ISSN: 2692-8205
Pages: 2026.08.06.742982
Section: New Results},
}

@article{clancy_na_2002,
	title = {Na+ {Channel} {Mutation} {That} {Causes} {Both} {Brugada} and {Long}-{QT} {Syndrome} {Phenotypes}},
	volume = {105},
	url = {https://www.ahajournals.org/doi/full/10.1161/hc1002.105183},
	doi = {10.1161/hc1002.105183},
	number = {10},
	
	journal = {Circulation},
	author = {Clancy, Colleen E. and Rudy, Yoram},
	month = mar,
	year = {2002},
	comment = {Publisher American Heart Association},
	pages = {1208--1213},
}

@book{adrian_bondy_graph_2008,
	edition = {3rd},
	series = {Graduate Texts in Mathematics},
    number={244},
	title = {Graph Theory},
	isbn = {1-84628-969-6 978-1-84628-969-9},
	publisher = {Springer},
	author = {Bondy, Adrian and Murty, U.S.R. },
	year = {2008},
    volume={244},
    address={London},
}

@book{kelly_reversibility_1981,
	title = {Reversibility and stochastic networks},
	author = {Kelly, F.P.},
	month = jun,
	year = {1981},
   publisher={Cambridge University Press},
   address={Cambridge, United Kingdom}
}

@article{lei_rapid_2019-1,
	title = {Rapid Characterization of {hERG} Channel Kinetics {I}: {Using} an Automated High-Throughput System},
	volume = {117},
	issn = {0006-3495},
	shorttitle = {Rapid {Characterization} of {hERG} {Channel} {Kinetics} {I}},
	url = {https://www.sciencedirect.com/science/article/pii/S0006349519305971},
	doi = {10.1016/j.bpj.2019.07.029},
	language = {en},
	number = {12},
	
	journal = {Biophysical Journal},
	author = {Lei, Chon Lok and Clerx, Michael and Gavaghan, David J. and Polonchuk, Liudmila and Mirams, Gary R. and Wang, Ken},
	month = dec,
	year = {2019},
	pages = {2438--2454},
	}

@article{hodgkin_quantitative_1952-1,
	title = {A quantitative description of membrane current and its application to conduction and excitation in nerve},
	volume = {117},
	issn = {1469-7793},
	url = {https://physoc.onlinelibrary.wiley.com/doi/abs/10.1113/jphysiol.1952.sp004764},
	doi = {10.1113/jphysiol.1952.sp004764},
	language = {en},
	number = {4},
	
	journal = {The Journal of Physiology},
	author = {Hodgkin, A. L. and Huxley, A. F.},
	year = {1952},
	comment = {\_eprint: https://physoc.onlinelibrary.wiley.com/doi/pdf/10.1113/jphysiol.1952.sp004764},
	pages = {500--544},
}

@article{li_improving_2017,
	title = {Improving the In Silico Assessment of Proarrhythmia Risk by Combining {hERG} ({Human} {Ether}-à-go-go-{Related} {Gene}) Channel-Drug Binding Kinetics and Multichannel Pharmacology},
	volume = {10},
	url = {https://www.ahajournals.org/doi/10.1161/CIRCEP.116.004628},
	doi = {10.1161/CIRCEP.116.004628},
	number = {2},
	
	journal = {Circulation: Arrhythmia and Electrophysiology},
	author = {Li, Zhihua and Dutta, Sara and Sheng, Jiansong and Tran, Phu N. and Wu, Wendy and Chang, Kelly and Mdluli, Thembi and Strauss, David G. and Colatsky, Thomas},
	month = feb,
	year = {2017},
	comment = {Publisher American Heart Association},
	pages = {e004628},
}

@article{sanguinetti_herg_2006,
	title = {{hERG} potassium channels and cardiac arrhythmia},
	volume = {440},
	issn = {1476-4687},
	doi = {10.1038/nature04710},
	language = {eng},
	number = {7083},
	journal = {Nature},
	author = {Sanguinetti, Michael C. and Tristani-Firouzi, Martin},
	month = mar,
	year = {2006},
	pmid = {16554806},
	pages = {463--469},
}

@book{Horn_Johnson_2012,
address={Cambridge, United Kingdom},
edition={2},
title={Matrix Analysis},
publisher={Cambridge University Press},
author={Horn, Roger A. and Johnson, Charles R.},
year={2012},
isbn= {9780521839402}
}

@article{lin_efficient_2018,
	title = {Efficient analysis of stochastic gene dynamics in the non-adiabatic regime using piecewise deterministic {Markov} processes},
	volume = {15},
	issn = {1742-5689},
	url = {https://doi.org/10.1098/rsif.2017.0804},
	doi = {10.1098/rsif.2017.0804},
	number = {138},
	urldate = {2026-05-08},
	journal = {Journal of The Royal Society Interface},
	author = {Lin, Yen Ting and Buchler, Nicolas E.},
	month = jan,
	year = {2018},
	pages = {20170804},
}

@article{wang_quantitative_1997,
	title = {A quantitative analysis of the activation and inactivation kinetics of {HERG} expressed in {Xenopus} oocytes.},
	volume = {502 ( Pt 1)},
	issn = {0022-3751 1469-7793},
	doi = {10.1111/j.1469-7793.1997.045bl.x},
	language = {eng},
	number = {Pt 1},
	journal = {The Journal of Physiology},
	author = {Wang, S. and Liu, S. and Morales, M. J. and Strauss, H. C. and Rasmusson, R. L.},
	month = jul,
	year = {1997},
	pmid = {9234196},
	pmcid = {PMC1159571},
	pages = {45--60},
}

@article{rudy_computational_2006,
	title = {Computational biology in the study of cardiac ion channels and cell electrophysiology},
	volume = {39},
	issn = {0033-5835},
	doi = {10.1017/S0033583506004227},
	language = {eng},
	number = {1},
	journal = {Quarterly Reviews of Biophysics},
	author = {Rudy, Yoram and Silva, Jonathan R.},
	month = feb,
	year = {2006},
	pmid = {16848931},
	pmcid = {PMC1994938},
	pages = {57--116},
}

@book{meyer_matrix_2000,
publisher = {Society for Industrial and Applied Mathematics},
booktitle = {Matrix analysis and applied linear algebra},
isbn = {0898714540},
year = {2000},
title = {Matrix analysis and applied linear algebra},
language = {eng},
address = {Philadelphia},
author = {Meyer, C. D.},
lccn = {00029725},
}

@incollection{colquhoun1995q,
	doi = {10.1007/978-1-4419-1229-9_20},
  title={A {Q-matrix} cookbook: how to write only one program to calculate the single-channel and macroscopic predictions for any kinetic mechanism},
  author={Colquhoun, David and Hawkes, Alan G},
  booktitle={Single-Channel Recording},
  editor={Sakmann, Bert and Neher, Erwin},
  edition={2nd},
  pages={589--633},
  year={1995},
  publisher={Springer},
  address={Boston, MA}
}

@article{gillespie_deterministic_2009,
	title = {Deterministic {Limit} of {Stochastic} {Chemical} {Kinetics}},
	volume = {113},
	issn = {1520-6106},
	url = {https://doi.org/10.1021/jp806431b},
	doi = {10.1021/jp806431b},
	number = {6},
	
	journal = {The Journal of Physical Chemistry B},
	author = {Gillespie, Daniel T.},
	month = feb,
	year = {2009},
	comment = {Publisher American Chemical Society},
	pages = {1640--1644},
}

@article{battelli_pseudo-steady-state_1985,
	title = {On the pseudo-steady-state approximation and {Tikhonov} theorem for general enzyme systems},
	volume = {75},
	issn = {0025-5564},
	url = {https://www.sciencedirect.com/science/article/pii/0025556485900392},
	doi = {10.1016/0025-5564(85)90039-2},
	number = {2},
	urldate = {2026-08-14},
	journal = {Mathematical Biosciences},
	author = {Battelli, Flaviano and Lazzari, Claudio},
	month = aug,
	year = {1985},
	pages = {229--246}
}

@article{moler2003nineteen,
	doi = {10.1137/S00361445024180},
  title={Nineteen dubious ways to compute the exponential of a matrix, twenty-five years later},
  author={Moler, Cleve and Van Loan, Charles},
  journal={SIAM review},
  volume={45},
  number={1},
  pages={3--49},
  year={2003},
  publisher={SIAM}
}

@article{teed_computationally_2016-1,
	title = {A computationally efficient algorithm for fitting ion channel parameters},
	volume = {3},
	issn = {2215-0161},
	url = {https://www.sciencedirect.com/science/article/pii/S2215016116300395},
	doi = {10.1016/j.mex.2016.11.001},
	
	journal = {MethodsX},
	author = {Teed, Zachary R. and Silva, Jonathan R.},
	month = jan,
	year = {2016},
	pages = {577--588},
}

@article{lau_auxin_2011,
	title = {Auxin triggers a genetic switch},
	volume = {13},
	copyright = {2011 Springer Nature Limited},
	issn = {1476-4679},
	url = {https://www.nature.com/articles/ncb2212},
	doi = {10.1038/ncb2212},
	language = {en},
	number = {5},
	
	journal = {Nature Cell Biology},
	author = {Lau, Steffen and Smet, Ive De and Kolb, Martina and Meinhardt, Hans and Jürgens, Gerd},
	month = may,
	year = {2011},
	pages = {611--615},
}

@article{luo_dynamic_1994,
	title = {A dynamic model of the cardiac ventricular action potential. {I}. {Simulations} of ionic currents and concentration changes},
	volume = {74},
	issn = {0009-7330},
	doi = {10.1161/01.res.74.6.1071},
	language = {eng},
	number = {6},
	journal = {Circulation Research},
	author = {Luo, C. H. and Rudy, Y.},
	month = jun,
	year = {1994},
	pmid = {7514509},
	pages = {1071--1096},
}

@article{kar_control_2016,
	title = {Control of {NFAT} {Isoform} {Activation} and {NFAT}-{Dependent} {Gene} {Expression} through {Two} {Coincident} and {Spatially} {Segregated} {Intracellular} {Ca2}+ {Signals}},
	volume = {64},
	issn = {1097-4164},
	doi = {10.1016/j.molcel.2016.11.011},
	language = {eng},
	number = {4},
	journal = {Molecular Cell},
	author = {Kar, Pulak and Mirams, Gary R. and Christian, Helen C. and Parekh, Anant B.},
	month = nov,
	year = {2016},
	pmid = {27863227},
	pmcid = {PMC5128683},
	pages = {746--759},
}

@article{bridge2012distinguishing,
	doi = {10.1016/j.mbs.2011.10.005},
  title={Distinguishing possible mechanisms for auxin-mediated developmental control in Arabidopsis: models with two Aux/IAA and ARF proteins, and two target gene-sets},
  author={Bridge, Lloyd. J. and Mirams, G. R. and Kieffer, M. L. and King, J. R. and Kepinski, S},
  journal={Mathematical biosciences},
  volume={235},
  number={1},
  pages={32--44},
  year={2012},
  publisher={Elsevier}
}

@book{hirsh_differential_2013,
editor = {Morris W. Hirsch and Stephen Smale and Robert L. Devaney},
title = {Differential Equations, Dynamical Systems, and an Introduction to Chaos},
publisher = {3rd},
edition = {Third Edition},
address = {Boston},
year = {2013},
isbn = {978-0-12-382010-5},
doi = {https://doi.org/10.1016/B978-0-12-382010-5.00055-5},
}

@article{leyser_auxin_2018,
	title = {Auxin {Signaling}},
	volume = {176},
	issn = {0032-0889},
	url = {https://doi.org/10.1104/pp.17.00765},
	doi = {10.1104/pp.17.00765},
	number = {1},
	
	journal = {Plant Physiology},
	author = {Leyser, Ottoline},
	month = jan,
	year = {2018},
	pages = {465--479},
}

@article{middleton_mathematical_2010,
	title = {Mathematical modelling of the {Aux}/{IAA} negative feedback loop},
	volume = {72},
	issn = {1522-9602},
	doi = {10.1007/s11538-009-9497-4},
	language = {eng},
	number = {6},
	journal = {Bulletin of Mathematical Biology},
	author = {Middleton, A. M. and King, J. R. and Bennett, M. J. and Owen, M. R.},
	month = aug,
	year = {2010},
	pmid = {20135237},
	pages = {1383--1407},
}

@article{farcot_modular_2015,
	title = {A {Modular} {Analysis} of the {Auxin} {Signalling} {Network}},
	volume = {10},
	issn = {1932-6203},
	url = {https://dx.plos.org/10.1371/journal.pone.0122231},
	doi = {10.1371/journal.pone.0122231},
	language = {en},
	number = {3},
	
	journal = {PLOS ONE},
	author = {Farcot, Etienne and Lavedrine, Cyril and Vernoux, Teva},
	editor = {Merks, Roeland M.H.},
	month = mar,
	year = {2015},
	pages = {e0122231},
}

@article{klipp_mathematical_2006,
	title = {Mathematical modeling of intracellular signaling pathways},
	volume = {7},
	issn = {1471-2202},
	url = {https://doi.org/10.1186/1471-2202-7-S1-S10},
	doi = {10.1186/1471-2202-7-S1-S10},
	language = {en},
	number = {1},
	
	journal = {BMC Neuroscience},
	author = {Klipp, Edda and Liebermeister, Wolfram},
	month = oct,
	year = {2006},
	pages = {S10},
}

@article{shuttleworth_evaluating_2025,
	title = {Evaluating the predictive accuracy of ion-channel models using data from multiple experimental designs},
	volume = {383},
	url = {https://royalsocietypublishing.org/doi/10.1098/rsta.2024.0211},
	doi = {10.1098/rsta.2024.0211},
	number = {2292},
	
	journal = {Philosophical Transactions of the Royal Society A: Mathematical, Physical and Engineering Sciences},
	author = {Shuttleworth, Joseph G. and Lei, Chon Lok and Windley, Monique J. and Hill, Adam P. and Preston, Simon P. and Mirams, Gary R.},
	month = mar,
	year = {2025},
	pages = {20240211},
}

@book{seneta06,
  address = {New York},
  author = {Seneta, E.},
  isbn = {0387297650},
  publisher = {Springer},
  refid = {209916821},
  title = {Non-negative matrices and Markov chains},
  year = 2006
}
\clearpage

\end{document}